\documentclass[11pt]{article}

\usepackage[boxed,nofillcomment,linesnumbered,noend,noresetcount, noline]{algorithm2e}
\usepackage{times}
\usepackage{fullpage}
\usepackage{cite}
\usepackage{graphicx}
\usepackage{epsf}
\usepackage{tabularx}
\usepackage{multirow}
\usepackage[font = footnotesize, center]{caption}
\usepackage{amsmath}
\usepackage{amsthm}
\usepackage{amssymb,latexsym,color}
\usepackage{mdwlist}
\usepackage{comment}

\renewcommand{\paragraph}[1]{ \vspace{0.1cm} \noindent\textbf{#1}\hspace{0.1em}}

\newcommand{\idlow}[1]{\mathord{\mathcode`\-="702D\it #1\mathcode`\-="2200}}
\newcommand{\id}[1]{\ensuremath{\idlow{#1}}}
\newcommand{\litlow}[1]{\mathord{\mathcode`\-="702D\sf #1\mathcode`\-="2200}}
\newcommand{\lit}[1]{\ensuremath{\litlow{#1}}}

\newtheorem{lemma}{Lemma}
\newtheorem{theorem}{Theorem}
\newtheorem{claim}{Claim}
\newtheorem{corollary}{Corollary}

\newenvironment{proofT}[1]{\proof\def\toto{#1}}{\hspace*{\fill}$\Box_{Theorem~\ref{\toto}}$\par\vspace{3mm}}

\newenvironment{proofL}[1]{\proof\def\toto{#1}}{\hspace*{\fill}$\Box_{Lemma~\ref{\toto}}$\par\vspace{3mm}}
\newenvironment{proofC}{\proof}{\hspace*{\fill}$\Box_{Corollary~\ref{\toto}}$\par\vspace{3mm}}

\newcommand{\thmpostponed}[2]
{
\newcounter{#1}
\setcounter{#1}{\value{theorem}}
\begin{theorem}
\label{#1}
#2
\end{theorem}
\expandafter\def\csname #1\endcsname{
\newcounter{#1temp}
\setcounter{#1temp}{\value{theorem}}
\setcounter{theorem}{\value{#1}}
\begin{theorem}
#2
\end{theorem}
\setcounter{theorem}{\value{#1temp}}
}
\vspace{-0em}
}

\newcommand{\thmpostponedwname}[3]
{
\newcounter{#1}
\setcounter{#1}{\value{theorem}}
\begin{theorem}[#2]
\label{thm:#1}
#3
\end{theorem}
\expandafter\def\csname #1\endcsname{
\newcounter{#1temp}
\setcounter{#1temp}{\value{theorem}}
\setcounter{theorem}{\value{#1}}
\begin{theorem}[#2]
#3
\end{theorem}
\setcounter{theorem}{\value{#1temp}}
}
\vspace{-1.5em}
}

\newcommand{\lemmaproofpostponedwname}[4]
{
\newcounter{#1}
\newcounter{#1temp}
\setcounter{#1}{\value{lemma}}
\begin{lemma}[#2]
\label{#1}
#3
\end{lemma}
\expandafter\def\csname #1\endcsname{
\setcounter{#1temp}{\value{lemma}}
\setcounter{lemma}{\value{#1}}
\begin{lemma}
#3
\end{lemma}
\setcounter{theorem}{\value{#1temp}}
\begin{proofL}{#1}
#4
\end{proofL}
}
}

\newcommand{\thmproofpostponedwname}[4]
{
\newcounter{#1}
\newcounter{#1temp}
\setcounter{#1}{\value{theorem}}
\begin{theorem}[#2]
\label{#1}
#3
\end{theorem}
\expandafter\def\csname #1\endcsname{
\setcounter{#1temp}{\value{theorem}}
\setcounter{theorem}{\value{#1}}
\begin{theorem}
#3
\end{theorem}
\setcounter{theorem}{\value{#1temp}}
\begin{proofT}{#1}
#4
\end{proofT}
}
}

\newcommand{\lemmapostponed}[2]
{
\newcounter{#1}
\newcounter{#1temp}
\setcounter{#1}{\value{theorem}}
\begin{lemma}
\label{#1}
#2
\end{lemma}
\expandafter\def\csname #1\endcsname{
\setcounter{#1temp}{\value{theorem}}
\setcounter{theorem}{\value{#1}}
\begin{lemma}
#2
\end{lemma}
\setcounter{theorem}{\value{#1temp}}
}
\vspace{-1.5em}
}

\newcommand{\lemmapostponedwname}[3]
{
\newcounter{#1}
\newcounter{#1temp}
\setcounter{#1}{\value{theorem}}
\begin{lemma}[#2]
\label{#1}
#3
\end{lemma}
\expandafter\def\csname #1\endcsname{
\setcounter{#1temp}{\value{theorem}}
\setcounter{theorem}{\value{#1}}
\begin{lemma}[#2]
#3
\end{lemma}
\setcounter{theorem}{\value{#1temp}}
}
\vspace{-1.5em}
}

\newcommand{\lemmaproofpostponed}[3]
{
\newcounter{#1}
\newcounter{#1temp}
\setcounter{#1}{\value{lemma}}
\begin{lemma}
\label{#1}
#2
\end{lemma}
\expandafter\def\csname #1\endcsname{
\setcounter{#1temp}{\value{lemma}}
\setcounter{lemma}{\value{#1}}
\begin{lemma}
#2
\end{lemma}
\setcounter{lemma}{\value{#1temp}}
\begin{proofL}{#1}
#3
\end{proofL}
}
}

\newcommand{\corollaryproofpostponed}[3]
{
\newcounter{#1}
\newcounter{#1temp}
\setcounter{#1}{\value{theorem}}
\begin{corollary}
\label{#1}
#2
\end{corollary}
\expandafter\def\csname #1\endcsname{
\setcounter{#1temp}{\value{theorem}}
\setcounter{theorem}{\value{#1}}
\begin{corollary}
#2
\end{corollary}
\setcounter{theorem}{\value{#1temp}}
\begin{proofC}
#3
\renewcommand{\toto}{#1}
\end{proofC}
}
}

\newcommand{\toto}{xxx}

\SetFuncSty{textsf}
\SetFuncSty{small}

\newtheorem{definition}{Definition} 

\newcommand{\todo}[1]{\typeout{TODO: #1}\textbf{[[[ #1 ]]]}}

\begin{document}

\title{Are Lock-Free Concurrent Algorithms Practically Wait-Free?}

\date{}

 \author{ Dan Alistarh \\ {\small MIT\thanks{alistarh@csail.mit.edu.}}
 \and Keren Censor-Hillel \\ {\small Technion\thanks{ckeren@cs.technion.ac.il. Shalon Fellow.}}
 \and Nir Shavit \\ {\small MIT \& Tel-Aviv University\thanks{shanir@csail.mit.edu.}} }

\maketitle

\begin{abstract}
Lock-free concurrent algorithms guarantee that \emph{some} concurrent operation will always make progress in a finite number of steps.
Yet programmers prefer to treat concurrent code as if it were \emph{wait-free}, guaranteeing that \emph{all} operations always make progress.
Unfortunately, designing wait-free algorithms is generally a very complex task, and the resulting algorithms are not always efficient.
While obtaining efficient wait-free algorithms has been a long-time goal for the theory community, most non-blocking commercial code is only lock-free.

This paper suggests a simple solution to this problem. 
We show that, for a large class of lock-free algorithms, under scheduling conditions which approximate those found in commercial hardware architectures, 
lock-free algorithms behave as if they are wait-free. In other words, programmers can keep on designing simple lock-free algorithms instead of complex wait-free ones, and in practice, they will get wait-free progress.


Our main contribution is a new way of analyzing a general class of lock-free algorithms under a \emph{stochastic scheduler}.
Our analysis relates the individual performance of processes with the global performance of the system using \emph{Markov chain lifting} between a
complex per-process chain and a simpler system progress chain.
We show that lock-free algorithms are not only wait-free with probability $1$,
but that in fact a general subset of lock-free algorithms can be closely bounded in terms of the average number of steps required until an operation completes. 

To the best of our knowledge, this is the first attempt to analyze progress conditions,
typically stated in relation to a worst case adversary, in a stochastic model capturing their expected asymptotic behavior.

\end{abstract}

\section{Introduction}

\newcommand{\SCULong}{Single-Compare-and-Swap Universal}
\newcommand{\SCU}{SCU}

The introduction of multicore architectures as today's main computing platform has brought about a renewed interest in concurrent data structures and algorithms, and
a considerable amount of research has focused on their modeling, design and analysis.

The behavior of concurrent algorithms is captured by  \emph{safety properties}, which guarantee their correctness,
and \emph{progress properties}, which guarantee their termination.
Progress properties can be quantified using two main criteria.
The first is whether the algorithm is \emph{blocking} or \emph{non-blocking}, that is, whether the delay of a single process will cause others to be blocked,
preventing them from terminating. Algorithms that use locks are blocking, while algorithms that do not use locks are non-blocking.
Most of the code in the world today is lock-based, though the fraction of code without locks is steadily growing~\cite{HSBook}.

The second progress criterion, and the one we will focus on in this paper, is whether a concurrent algorithm guarantees \emph{minimal} or \emph{maximal progress} \cite{HS11}.
Intuitively, minimal progress means that \emph{some} process is always guaranteed to make progress by completing its operations,
while maximal progress means that \emph{all} processes always complete all their operations.

Most non-blocking commercial code is \emph{lock-free}, that is, provides minimal progress without using locks~\cite{HS11, FraserThesis}. 
Most blocking commercial code is \emph{deadlock-free}, that is, provides minimal progress when using locks. 
Over the years, the research community has devised ingenious, technically sophisticated algorithms that provide \emph{maximal progress}:
such algorithms are either \emph{wait-free}, i.e. provide maximal progress without using locks\cite{Her91}, or \emph{starvation-free}\cite{Lamport74a}, 
i.e. provide maximal progress when using locks.
Unexpectedly, maximal progress algorithms, and wait-free algorithms in particular, are not being adopted by practitioners, 
despite the fact that the completion of all method calls in a program is a natural assumption that programmers implicitly make.

Recently, Herlihy and Shavit \cite{HS11} suggested that perhaps the answer lies in a surprising property of lock-free algorithms: in practice, they often behave as if they were wait-free (and similarly, deadlock-free algorithms behave as if they were starvation-free). Specifically, most operations complete in a timely manner, and the impact of long worst-case executions on performance
is negligible. In other words, in real systems, the scheduler that governs the threads' behavior in long executions does not single out any particular thread in order to cause the theoretically possible bad behaviors.
This raises the following question: 
could the choice of \emph{wait-free} versus \emph{lock-free} be based simply on what assumption a programmer is willing to make about the underlying scheduler, 
and, with the right kind of scheduler, one will not need wait-free algorithms except in very rare cases? 

This question is important because the difference between a wait-free and a lock-free algorithm for any given problem typically involves the introduction of 
specialized ``helping'' mechanisms~\cite{Her91}, which significantly increase the complexity (both the design complexity and time complexity) of the solution. If one could simply rely on the scheduler, 
adding a helping mechanism to guarantee wait-freedom (or starvation-freedom) would be unnecessary.

Unfortunately, there is currently no analytical framework which would allow answering the above question, since it would require predicting
the behavior of a concurrent algorithm over long executions, under a scheduler that is not adversarial.

\vspace{.3em}
\paragraph{Contribution.} In this paper, we take a first step towards such a framework. Following empirical observations, we introduce a  \emph{stochastic scheduler} model,
and use this model to predict the long-term behavior of a general class of concurrent algorithms. 
The stochastic scheduler is similar to an adversary: at each time step, it picks some process to schedule.
The main distinction is that, in our model, the scheduler's choices \emph{contain some randomness}.
In particular, a stochastic scheduler has a probability threshold $\theta > 0$ such that every
(non-faulty) process is scheduled with probability at least $\theta$ in each step.

We start from the following observation:
under \emph{any stochastic scheduler}, every \emph{bounded lock-free} algorithm is actually \emph{wait-free with probability $1$}.
(A \emph{bounded} lock-free algorithm guarantees that \emph{some} process always makes progress within a finite progress bound.)
In other words, for any such algorithm, the schedules which prevent a process from ever making progress must have probability mass $0$.
The intuition is that, with probability $1$, each specific process eventually takes enough consecutive steps, implying that it completes its operation.
This observation generalizes to any bounded minimal/maximal progress condition~\cite{HS11}: we show that 
under a stochastic scheduler, bounded minimal progress becomes maximal progress, with probability $1$. 
However, this intuition is insufficient for explaining why lock-free data structures are \emph{efficient} in practice: because it works for arbitrary algorithms,
the upper bound it yields on the number of steps until an operation completes is unacceptably high. 

Our main contribution is analyzing a general class of lock-free algorithms under a specific stochastic scheduler, and showing that not only are they wait-free with probability $1$,
but that in fact they provide a pragmatic bound on the number of steps until each operation completes. 

We address a refined \emph{uniform} stochastic scheduler,
which schedules each non-faulty process with uniform probability in every step.
Empirical data suggests that, in the long run, the uniform stochastic scheduler is a reasonable approximation
for a real-world scheduler (see Figures~\ref{fig:average} and~\ref{fig:steps}). 
We emphasize that we do not claim real schedulers are uniform stochastic, 
but only that such a scheduler gives a good approximation of what happens in practice for our complexity measures, over long executions. 

We call the algorithmic class we analyze \emph{single compare-and-swap universal} (SCU). 
An algorithm in this class is divided into a \emph{preamble}, and a \emph{scan-and-validate} phase.
The preamble executes auxiliary code, such as local updates and memory allocation.
In the second phase, the process first determines the data structure state by scanning the memory.
It then locally computes the updated state after its method call would be performed,
and attempts to commit this state to memory by performing an atomic \emph{compare-and-swap} (CAS) operation.
If the CAS operation succeeds, then the state has been updated, and the method call completes.
Otherwise, if some other process changes the state in between the scan and the attempted update,
then the CAS operation fails, and the process must restart its operation.

This algorithmic class is widely used to design lock-free data structures.
It is known that every sequential object has a lock-free implementation in this class using a lock-free version of Herlihy's universal construction~\cite{Her91}.
Instances of this class are used to obtain efficient data structures such as stacks~\cite{Treiber}, queues~\cite{MichaelS96}, or 
hash tables~\cite{FraserThesis}.
The read-copy-update (RCU)~\cite{RCU} synchronization mechanism employed by the Linux kernel is also an instance of this pattern.

We examine the  class $\SCU{}$ under a uniform stochastic scheduler,
and first observe that, in this setting, every such algorithm behaves as a Markov chain.
The computational cost of interest is \emph{system steps}, i.e.~shared memory accesses by the processes.
The complexity metrics we analyze are \emph{individual latency}, which is
the expected number of steps of the system until a specific process completes a method call,
and \emph{system latency}, which is the expected number of steps of the system
to complete \emph{some} method call.
We bound these parameters
by studying the stationary distribution of the Markov chain induced by the algorithm.

We prove two main results. The first is that, in this setting, all algorithms in this class have the property that the individual
latency of any process is $n$ times the system latency.
In other words, the expected number of steps for any two processes to complete an operation is \emph{the same};
moreover, the expected number of steps for the system to complete any operation is the expected number of steps for a specific process
to complete an operation, divided by $n$.
The second result is an upper bound of $O( q + s \sqrt n )$ on the system latency, where $q$ is the number of steps in the preamble,
$s$ is the number of steps in the scan-and-validate phase, and $n$ is the number of processes. This bound is asymptotically tight. 

The key mathematical tool we use is \emph{Markov chain lifting}~\cite{ChenLP99, HayesSinclair}.
More precisely, for such algorithms, we prove that
there exists a function which \emph{lifts} the complex Markov chain induced by the algorithm to a simplified \emph{system} chain.
The asymptotics of the system latency can be determined directly from the minimal progress chain. 
In particular, we bound system latency by characterizing the behavior of a new type of \emph{iterated} balls-into-bins game, 
consisting of iterations which end when a certain condition on the bins first occurs, after which some of the bins change their state and a new iteration begins. 
Using the lifting, we prove that the individual latency is always $n$ times the system latency.

In summary, our analysis shows that, under an approximation of the real-world scheduler, 
a large class of lock-free algorithms provide virtually the same progress guarantees as wait-free ones, and that, roughly, 
the system completes requests at a rate that is $n$ times that of individual processes.  
More generally, it provides for the first time an analytical framework for predicting the behavior of a class of concurrent algorithms, over long executions, under a scheduler that is not adversarial.

\vspace{.3em}
\paragraph{Related work.} To the best of our knowledge, the only prior work which addresses a probabilistic scheduler for a shared memory environment 
is that of Aspnes~\cite{Aspnes2002}, who gave a fast  consensus algorithm  under a probabilistic scheduler model different from the one considered in this paper. 
The observation that many lock-free algorithms behave as wait-free in practice was made by Herlihy and Shavit in the context of formalizing 
minimal and maximal progress conditions~\cite{HS11}, and is well-known among practitioners. 
For example, reference~\cite[Figure $6$]{AlBahra} gives empirical results for the latency distribution of individual operations of a lock-free stack. 
Recent work by Petrank and Timnat~\cite{PT14} states that most known lock-free algorithms can be written in a canonical form, 
which is similar to the class $\id{SCU}$, but more complex than the pattern we consider. 
Significant research interest has been dedicated to transforming obstruction-free or lock-free algorithms to wait-free ones, e.g.~\cite{PT14, Kogan}, 
while minimizing performance overhead. In particular, an efficient strategy has been to divide the algorithm into a lock-free \emph{fast path}, and a wait-free 
\emph{backup path}, which is invoked it an operation fails repeatedly. 
Our work does not run contrary to this research direction, since the progress guarantees we prove are only probabilistic. Instead, 
it could be used to bound the cost of the backup path during the execution.

\vspace{.3em}
\paragraph{Roadmap.} We describe the model, progress guarantees, and complexity metrics in Section~\ref{sec:model}.  
In particular, Section~\ref{sec-stochastic-scheduler} defines the stochastic scheduler. 
We show that minimal progress becomes maximal progress with probability $1$ in Section~\ref{sec:LFtoWF}. 
Section~\ref{sec:algorithm} defines the class $\SCU(q, s)$, while Section~\ref{scan-validate} analyzes individual and global latency. 
The Appendix contains empirical justification for the model, and a comparison between the predicted 
behavior of an algorithm  and its practical performance.


\section{System Model}
\label{sec:model}

\subsection{Preliminaries}

\paragraph{Processes and Objects.} We consider a shared-memory model,
in which $n$ processes $p_1, \ldots, p_{n}$, communicate through registers, on which they perform atomic \lit{read}, \lit{write}, and \lit{compare-and-swap} (CAS) operations.
A CAS operation takes three arguments $(R, \id{expVal}, \id{newVal})$, where $R$ is the register on which it is applied, $\id{expVal}$ is the expected value of the register,
and $\id{newVal}$ is the new value to be written to the register. If $\id{expVal}$ matches the value of $R$,
then we say that the CAS is successful, and the value of $R$ is updated to $\id{newVal}$. Otherwise, the CAS fails.
The operation returns \emph{true} if it successful, and \emph{false} otherwise.

We assume that each process has a unique identifier. Processes follow an algorithm,  composed of
shared-memory steps and local computation.
The order of process steps is controlled by the \emph{scheduler}.
A set of at most $n - 1$ processes may fail by crashing. 
A crashed process stops taking steps for the rest of the execution. 
A process that is not crashed at a certain step is \emph{correct}, 
and if it never crashes then it takes an infinite number of steps in the execution.

The algorithms we consider are implementations of shared objects.
A shared object $O$ is an abstraction providing a set of \emph{methods} $M$,
each given by its sequential specification.
In particular, an implementation of a method $m$ for object $O$ is a set of $n$ algorithms,
one for each executing process.
When process $p_i$ invokes method $m$ of object $O$,
it follows the corresponding algorithm until it receives a response from the algorithm.
In the following, we do not distinguish between a method $m$ and its implementation.
A method invocation is \emph{pending} if has not received a response.
A method invocation is \emph{active} if it is made by a \emph{correct} process (note that the process may still crash in the future).


\paragraph{Executions, Schedules, and Histories.}
An execution is a sequence of operations performed by the processes.
To represent executions, we assume discrete time,
where at every time unit only one process is scheduled.
In a time unit, a process can perform any number of local computations or coin flips,
after which it issues a \emph{step}, which consists of a single shared memory operation.
Whenever a process becomes active, as decided by the scheduler, it performs its local computation and then executes a step.
The \emph{schedule} is a (possibly infinite) sequence of process identifiers.
If process $p_i$ is in position $\tau \geq 1$ in the sequence, then $p_i$ is active at time step $\tau$.

Raising the level of abstraction, we define a \emph{history} as a finite
sequence of method invocation and response events. Notice that each schedule has a corresponding history, in which
individual process steps are mapped to method calls. On the other hand, a history can be the image of several schedules.

\subsection{Progress Guarantees}

We now define minimal and maximal progress guarantees. 
We partly follow the unified presentation from~\cite{HS11}, except that we do not specify progress guarantees for each method of an object. 
Rather, for ease of presentation, we adopt the simpler definition which specifies progress provided by an implementation. 
Consider an execution $e$, with the corresponding history $H_e$.
An implementation of an object $O$ provides \emph{minimal progress} in the execution $e$ if,
in every suffix of $H_e$, some pending active instance of some method has a matching response.
Equivalently, there is no point in the corresponding execution from which all the processes take an infinite number of steps without returning from their invocation.

An implementation provides \emph{maximal} progress in an execution $e$ if, in every suffix of the corresponding history $H_e$, \emph{every}
pending active invocation of a method has a response.
Equivalently, there is no point in the execution from which a process takes infinitely many steps without returning.

\paragraph{Scheduler Assumptions.}
We say that an execution is \emph{crash-free} if each process is always correct, i.e. if each process takes an infinite number of steps.
An execution is \emph{uniformly isolating} if, for every $k > 0$, every correct process has an interval where it takes at least $k$ consecutive steps.

\paragraph{Progress.} An implementation is \emph{deadlock-free} if it guarantees minimal progress in every crash-free execution, and maximal progress in some crash-free execution.\footnote{According to~\cite{HS11}, the algorithm is required to guarantee maximal progress in some execution to rule out pathological cases where a thread locks the object and never
releases the lock.}
An implementation is \emph{starvation-free} if it guarantees maximal progress in every crash-free execution.
An implementation is \emph{clash-free} if it guarantees minimal progress in every uniformly isolating history, and maximal progress in some such history~\cite{HS11}.
An implementation is \emph{obstruction-free} if it guarantees maximal progress in every uniformly isolating execution\footnote{This is the definition of obstruction freedom from~\cite{HS11}; it is weaker than the one in~\cite{HerlihyLM03} since it assumes uniformly isolating schedules only, but we use it here as it complies with our requirements of providing maximal progress.}.
An implementation is \emph{lock-free} if it guarantees minimal progress in every execution, and maximal progress in some execution.
An implementation is \emph{wait-free} if it guarantees maximal progress in every execution.

\paragraph{Bounded Progress.} While the above definitions provide reasonable measures of progress, often in practice more explicit progress guarantees may be desired, which provide an upper bound on the number of steps until some method makes progress.
To model this, we say that an implementation guarantees \emph{bounded minimal progress}
if there exists a bound $B > 0$ such that, for any time step $t$ in the execution $e$ at which there is an active invocation of some method,
some invocation of a method returns within the next $B$ steps by all processes.
An implementation guarantees \emph{bounded maximal progress} if there exists a bound
$B > 0$ such that \emph{every} active invocation of a method returns within $B$ steps by all processes.
We can specialize the  definitions of bounded progress guarantees
to the scheduler assumptions considered above to obtain definitions for \emph{bounded} deadlock-freedom, \emph{bounded} starvation-freedom, and so on.

%
%

\subsection{Stochastic Schedulers}
\label{sec-stochastic-scheduler}


We define a stochastic scheduler as follows.

\begin{definition}[Stochastic Scheduler]
For any $n \geq 0$, a scheduler for $n$ processes is defined by a triple $(\Pi_\tau, A_\tau, \theta)$.
The parameter $\theta \in [0, 1]$ is the \emph{threshold}.
For each time step $\tau \geq 1$, $\Pi_\tau$ is a probability distribution
for scheduling the $n$ processes at  $\tau$, and $A_\tau$ is the subset of \emph{possibly active} processes at time step $\tau$.
At time step $\tau \geq 1$, the distribution $\Pi_\tau$ gives, for every $i \in \{1, \ldots, n\}$
a probability $\gamma^i_\tau$, with which process $p_i$ is scheduled.
The distribution $\Pi_\tau$ may depend on arbitrary outside factors, such as the current state of the algorithm being scheduled.
A scheduler $(\Pi_\tau, A_\tau, \theta)$ is \emph{stochastic} if $\theta > 0$. For every $\tau \geq 1$, the parameters must ensure the following:
\begin{enumerate*}
 \item (Well-formedness) $\sum_{i = 1}^n \gamma^i_\tau = 1$;
 \item (Weak Fairness) For every process $p_i \in A_\tau$, $\gamma^i_\tau \geq \theta $;
 \item (Crashes) For every process $p_i \notin A_\tau$, $\gamma^i_\tau = 0$;
 \item (Crash Containment) $A_{\tau + 1} \subseteq A_\tau$.
\end{enumerate*}

\end{definition}

The well-formedness condition ensures that some process is always scheduled. Weak fairness ensures
that, for a stochastic scheduler, possibly active processes do get scheduled with some non-zero probability.
The crash condition ensures that failed processes do not get scheduled.
The set $\{p_1, p_2, \ldots, p_n\} \setminus A_\tau$ can be seen as the set of crashed processes at time step $\tau$,
since the probability of scheduling these processes at every subsequent time step is $0$.

\paragraph{An Adversarial Scheduler.} Any classic asynchronous shared memory adversary can be modeled by ``encoding''
its adversarial strategy in the probability distribution $\Pi_\tau$ for each step.
Specifically, given an algorithm $A$ and a worst-case adversary $\mathcal{A}_A$ for $A$, let $p_i^\tau$ be the process that
is scheduled by $\mathcal{A}_A$ at time step $\tau$.
Then we give probability $1$ in $\Pi_\tau$ to process $p_i^\tau$, and $0$ to all other processes.
Things are more interesting when the threshold $\theta$ is strictly more than $0$,
i.e., there is some randomness in the scheduler's choices. 

\paragraph{The Uniform Stochastic Scheduler.}
A natural scheduler is the \emph{uniform} stochastic scheduler, for which,  
assuming no process crashes, we have that $\Pi_\tau$ has $\gamma_i^\tau = 1 / n$, for all $i$ and $\tau \geq 1$,
and $A_\tau = \{1, \ldots, n\}$ for all time steps $\tau \geq 1$. With crashes, we have that $\gamma_i^{\tau}=1/|A_{\tau}|$ if $i \in A_{\tau}$, and $\gamma_i^{\tau}=0$ otherwise.

\subsection{Complexity Measures}

Given a concurrent algorithm, standard analysis focuses on two measures: \emph{step complexity},
the worst-case number of steps performed by a single process in order to return from a method invocation, and \emph{total step complexity}, or \emph{work},
which is the worst-case number of system steps required to complete invocations of all correct processes when performing a task together.
In this paper, we focus on the analogue of these complexity measures for long executions.
Given a stochastic scheduler, we define \emph{(average) individual latency} as the maximum over all inputs of the
expected number of steps taken by the system between the returns times of two consecutive invocations of the same process.
Similarly, we define the \emph{(average) system latency} as the maximum over all inputs of the expected number of system steps between consecutive returns times of any two invocations.

\section{Background on Markov Chains}
\label{sec:markov-chain-background}

We now give a brief overview of Markov chains. Our presentation follows standard texts, e.g.~\cite{Peres, Mitz}.
The definition and properties of Markov chain lifting are adapted from~\cite{HayesSinclair}.

Given a set $S$, a sequence of random variables $(X_t)_{t \in \mathbb{N}}$, where $X_t \in S$, is a (discrete-time) \emph{stochastic process} with states in $S$.
A \emph{discrete-time Markov chain} over the state set $S$ is a discrete-time stochastic process with states in $S$ that satisfies the \emph{Markov condition}
$$\Pr[X_t = i_t | X_{t - 1} = i_{t - 1}, \ldots, X_0 = i_0 ] = \Pr[ X_t = i_t | X_{t - 1} = i_{t - 1} ].$$

\noindent The above condition is also called the \emph{memoryless property}.
A Markov chain is \emph{time-invariant} if the equality $\Pr[ X_t = j | X_{t - 1} = i ] = \Pr[ X_{t'} = j | X_{t' - 1} = i]$ holds for all times $t, t' \in \mathbb{N}$
and all $i, j \in S$. This allows us to define the \emph{transition matrix} $P$ of a Markov chain as the matrix with entries
$$ p_{ij} = \Pr[ X_{t} = j | X_{t - 1} = i ].$$
\noindent The \emph{initial distribution} of a Markov chain is given by the probabilities $\Pr[X_0 = i]$, for all $i \in S$.
We denote the time-invariant Markov chain $X$ with initial distribution $\lambda$ and transition matrix $P$ by $M( P, \lambda)$.

The random variable $T_{ij} = \min\{ n \geq 1|X_n = j, \textnormal{ if } X_0 = i\}$ counts the number of steps needed by the Markov chain to get from $i$ to $j$,
and is called the \emph{hitting time} from $i$ to $j$. We set $T_{i, j} = \infty$ if state $j$ is unreachable from $i$.
Further, we define $h_{ij} = E[T_{ij}]$, and call $h_{ii} = E[T_{ii}]$ the (expected) return time for state $i \in S$.

Given $P$, the transition matrix of $M(P, \lambda)$, a \emph{stationary distribution}  of the Markov chain is a state vector $\pi$ with $\pi =  \pi P$.
(We consider \emph{row} vectors throughout the paper.)
The intuition is that if the state vector of the Markov chain is $\pi$ at time $t$, then it will remain $\pi$ for all $t' > t$.
Let $P^{(k)}$ be the transition matrix $P$ multiplied by itself $k$ times, and $p_{ij}^{(k)}$ be element $(i, j)$ of $P^{(k)}$.
A Markov chain is \emph{irreducible} if for all pairs of states $i, j \in S$ there exists $m \geq 0$ such that $p_{ij}^{(m)} > 0$.
(In other words, the underlying graph is strongly connected.)
This implies that $T_{ij} < \infty$, and all expectations $h_{ij}$ exist, for all $i, j \in S$. Furthermore, the following is known.

\begin{theorem}
\label{thm:stat}
 An irreducible finite Markov chain has a unique stationary distribution $\pi$, namely
 $$\pi_j = \frac{1}{h_{jj}}, \forall j \in S.$$
\end{theorem}

The periodicity of a state $j$ is the maximum positive integer $\alpha$ such that
$\{ n \in \mathbb{N} | p_{jj}^{(n)} > 0\} \subseteq \{i\alpha | i \in \mathbb{N}\}.$ A state with periodicity $\alpha = 1$ is called \emph{aperiodic}.
A Markov chain is \emph{aperiodic} if all states are aperiodic. If a Markov chain has at least one self-loop, then it is aperiodic.
A Markov chain that is irreducible and aperiodic is \emph{ergodic}.
Ergodic Markov chains converge to their stationary distribution as $t \rightarrow \infty$ independently of their initial distributions.

\begin{theorem}
For every ergodic finite Markov chain $(X_t)_{t \in \mathbb{N}}$ we have independently of the initial distribution that $\lim_{t\rightarrow \infty} q_t = \pi$,
where $\pi$ denotes the chain's unique stationary distribution, and $q_t$ is the distribution on states at time $t \in \mathbb{N}$. 
\end{theorem}

\paragraph{Ergodic Flow.} It is often convenient to describe an ergodic Markov chain in terms of its \emph{ergodic flow}: for each (directed) edge $ij$,
we associate a flow $Q_{ij} = \pi_i p_{ij}$. These values satisfy $\sum_i Q_{ij} = \sum_i Q_{ji}$ and $\sum_{i,j} Q_{ij} = 1$. It also holds that
$\pi_j = \sum_i Q_{ij}$.

\paragraph{Lifting Markov Chains.} Let $M$ and $M'$ be ergodic Markov chains on finite state spaces $S, S'$, respectively.
Let $P, \pi$ be the transition matrix and stationary distribution for $M$, and $P', \pi'$ denote the corresponding objects for $M'$.
We say that $M'$ is a \emph{lifting} of $M$~\cite{HayesSinclair} if there is a function $f : S' \rightarrow S$ such that
$$ Q_{ij} = \sum_{x \in f^{-1}(i), y \in f^{-1}(j)} Q'_{xy}, \forall i, j \in S.$$

Informally, $M'$ is collapsed onto $M$ by clustering several of its states into a single state, as specified by the function $f$.
The above relation specifies a homomorphism on the ergodic flows. An immediate consequence of this relation is the following
connection between the stationary distributions of the two chains.
\begin{lemma}
For all $v \in S$, we have that
  $$ \pi(v) = \sum_{x \in f^{-1} (v)} \pi'(x).$$
\end{lemma}


\section{From Minimal Progress to Maximal Progress}
\label{sec:LFtoWF}

We now formalize the intuition that, under a stochastic scheduler, all algorithms ensuring bounded minimal progress
guarantee in fact maximal progress with probability $1$.
We also show the \emph{bounded} minimal progress assumption is necessary: if minimal progress is not bounded, then maximal progress may not be achieved.

\begin{theorem}[Min to Max Progress]
 \label{minmax}
  Let $\mathcal{S}$ be a stochastic scheduler with probability threshold $1 \geq \theta > 0$.
 Let $A$ be an algorithm ensuring bounded minimal progress with a bound $T$. Then $A$ ensures maximal progress with probability $1$.
 Moreover, the expected maximal progress bound of $A$ is at most $(1 / \theta)^T$.
\end{theorem}
\begin{proof}
 Consider an interval of $T$ steps in an execution of algorithm $A$. Our first observation is that, since $A$ ensures $T$-bounded minimal progress,
 any process that performs $T$ \emph{consecutive} steps in this interval must complete a method invocation.
 To prove this fact, we consider cases on the minimal progress condition.
 If the minimal progress condition is $T$-bounded \emph{deadlock-freedom} or \emph{lock-freedom}, then every sequence of $T$ steps by the algorithm
 must complete some method invocation. In particular, $T$ steps by a single process must complete a method invocation. Obviously, this completed method invocation
 must be by the process itself. If the progress condition is $T$-bounded \emph{clash-freedom}, then the claim follows directly from the definition.

 Next, we show that, since $\mathcal{S}$ is a stochastic scheduler with positive probability threshold, each correct process will eventually be scheduled for
 $T$ consecutive steps, with probability $1$. By the weak fairness condition in the definition, for every time step $\tau$, every active process $p_i \in A_\tau$ is scheduled with probability
 at least $\theta > 0$. A process $p_i$ is \emph{correct} if $p_i \in A_\tau$, for all $\tau \geq 1$.
 By the definition, at each time step $\tau$, each correct process $p_i \in A_\tau$
 is scheduled for $T$ consecutive time units with probability at least $\theta^T > 0$. From the previous argument, it follows that
 every correct process eventually completes each of its method calls with probability $1$.
 By the same argument, the expected completion time for a process is at most $(1 / \theta)^{T}$.
\end{proof}

The proof is based on the fact that, for every correct process $p_i$, eventually, the scheduler will produce a solo
a schedule of length $T$. On the other hand, since the algorithm ensures minimal progress with bound $T$,
we show that $p_i$ must complete its operation
during this interval.


We then prove that the finite bound for minimal progress is necessary.
For this, we devise an \emph{unbounded} lock-free algorithm which is not wait-free with probability $> 0$.
The main idea is to have processes that fail to change the value of a CAS repeatedly increase the number of steps they need to take to complete an operation. 
(See Algorithm~\ref{alg:unboundedLF}.)

\setcounter{AlgoLine}{0}
\begin{algorithm}[t]
\textbf{Shared}: CAS object $C$, initially $0$\;
Register $R$\\
\textbf{Local}: Integers $v, val, j$, initially $0$\\
\While{\lit{true}}
{ 
\quad $val \leftarrow \lit{CAS}(C,v,v+1)$\\
\quad \lIf {$val = v$} {\textbf{return}\;}
\quad \Else {
\quad \quad $v \leftarrow val$\\
\quad \quad \lFor {$j=1 \dots n^2 v$} { $read(R)$ \;}
}
}

\caption{An unbounded lock-free algorithm.}
\label{alg:unboundedLF}
\end{algorithm}

\begin{lemma}
 \label{unbounded}
 There exists an unbounded lock-free algorithm that is \emph{not wait-free} with high probability.
\end{lemma}
\begin{proof}
Consider the initial state of Algorithm~\ref{alg:unboundedLF}.
With probability at least $1/n$, each process $p_i$ can be the first process to take a step,
performing a successful CAS operation. Assume process $p_1$ takes the first step.
Conditioned on this event,
let $P$ be the probability that $p_1$ is not the next process that performs a successful CAS operation.
If $p_1$ takes a step in any of the next $n^2\cdot v$ steps, then it is the next process that wins the CAS.
The probability that this does not happen is at most $(1-1/n)^{n^2}$.
Summing over all iterations,
the probability that $p_1$ ever performs an unsuccessful CAS is therefore at most $\sum_{\ell=1}^{\infty}{(1-1/n)^{n^2\cdot \ell}} \leq 2(1-1/n)^{n^2} \leq 2e^{-n}$.
Hence, with probability at least $1-2e^{-n}$, process $p_1$ always wins the CAS,
while other processes never do. This implies that the algorithm is not wait-free, with high probability.
\end{proof}

\section{The Class of Algorithms $\SCU(q, s)$}
\label{sec:algorithm}

\DontPrintSemicolon
\begin{algorithm}[t]
{\small
\textbf{Shared}: registers $R, R_1, R_2, \ldots, R_{s - 1}$\;
\Indp
	 \textbf{procedure} $\lit{method-call}()$\;
	 Take preamble steps $O_1, O_2, O_q$ 	 \tcc{Preamble region}
	 \While{ $\lit{true}$ }
	 {
	    \tcc{Scan region:}
	    $v \gets R.\lit{read}()$\;
	    $v_1 \gets R_1.\lit{read}()$;
	    $v_2 \gets R_2.\lit{read}()$;
	    $\ldots$;
	    $v_{s - 1} \gets R_{s - 1}.\lit{read}()$\;
	    $v' \gets $ new proposed state based on $v, v_1, v_2, \ldots, v_{s - 1}$\;
	    \tcc{Validation step:}
	    $\id{flag} \gets \lit{CAS}( R, v, v')$\;
	    \If{ $\id{flag} = \lit{true}$ }
	    {
	      \textbf{output} $\id{success}$\;
	    }
	 }
\Indm
}
\caption{The structure of the lock-free algorithms in $\SCU_{q,s}$.}
\label{fig:uc}
\end{algorithm}

In this section, we define the class of algorithms $\SCU(q, s)$.
An algorithm in this class is structured as follows. (See Algorithm~\ref{fig:uc} for the pseudocode.)
The first part is the \emph{preamble}, where the process performs a series of $q$ steps.
The algorithm then enters a \emph{loop},  divided into a \emph{scan} region, which reads the values of $s$ registers,
and a \emph{validation} step, where the process performs a CAS operation, which attempts to change the value of a register.
The of the scan region is to obtain a view of the data structure state. In the validation step,
the process checks that this state is still valid, and attempts to change it. 
If the CAS is successful, then the operation completes. Otherwise, the process restarts the loop.
We say that an algorithm with the above structure with parameters $q$ and $s$ is in $\SCU(q, s)$.

We assume that steps in the preamble may perform memory updates, including to registers $R_1, \ldots, R_{s - 1}$,
but do not change the value of the decision register $R$. 
Also, two processes never propose the same value for the register $R$. (This can be easily enforced by adding a timestamp to each request.) 
The order of steps in the scan region can be changed without affecting our analysis.
Such algorithms are used in several CAS-based concurrent implementations.
In particular, the class can be used to implement a concurrent version of every sequential object~\cite{Her91}.
It has also been used to obtain efficient implementations of several concurrent objects, such as fetch-and-increment~\cite{DiceLM13},
stacks~\cite{Treiber}, and queues~\cite{MichaelS96}.

\section{Analysis of the Class $\SCU(q, s)$}

We analyze the performance of algorithms in $\SCU(q, s)$ under the uniform stochastic scheduler. We assume that all threads execute the same 
method call with preamble of length $q$, and scan region of length $s$. 
Each thread executes an infinite number of such operations.
To simplify the presentation, we assume all $n$ threads are correct in the analysis. 
The claim is similar in the crash-failure case, and will be considered separately. 

We examine two parameters:
system latency, i.e., how often (in terms of system steps) does a new operation complete, and individual latency, i.e., how often
does \emph{a certain thread} complete a new operation.
Notice that the worst-case latency for the whole system is $\Theta(q+s n)$ steps,
while the worst-case latency for an individual thread is $\infty$, as the algorithm is not wait-free.
We will prove the following result:

\thmpostponed{thmscu}
{
Let $A$ be an algorithm in $\SCU(q, s)$. Then, under the uniform stochastic scheduler, the system latency of $A$ is $O( q + s\sqrt{n} )$, and the individual latency is $O( n ( q + s\sqrt{n} ) )$.
}

We prove the upper bound by splitting the class $\SCU(q,s)$ into two separate components, and analyzing each under the uniform scheduler.
The first part is the loop code, which we call the \emph{scan-validate} component. The second part is the \emph{parallel code},
which we use to characterize the performance of the preamble code. In other words, we first consider $\SCU(0,s)$ and then $\SCU(q,0)$.

\subsection{The Scan-Validate Component}
\label{scan-validate}

Notice that, without loss of generality, we can simplify the pseudocode to contain a single read step before the CAS. 
We obtain the performance bounds for this simplified algorithm,
and then multiply them by $s$, the number of scan steps. That is, we start by analyzing $\SCU(0, 1)$ and then generalize to $\SCU({0,s})$.

\paragraph{Proof Strategy.} We start from the Markov chain representation of the algorithm, which we call the \emph{individual chain}. 
We then focus on a simplified representation, which only tracks \emph{system-wide progress}, irrespective of 
which process is exactly in which state. We call this the \emph{system chain}. We first prove 
the individual chain can be related to the system chain via a lifting function,  
which allows us to relate the individual latency to the system latency (Lemma~\ref{lemlifting}). 
We then focus on bounding system latency. 
We describe the behavior of the system chain via an iterated balls-and-bins game, whose stationary behavior we analyze in Lemmas~\ref{phaselength} and~\ref{lemranges}. 
Finally, we put together these claims to obtain an $O( \sqrt n )$ upper bound on the system latency of $\SCU(0, 1)$.


\setcounter{AlgoLine}{0}
\begin{algorithm}[t]
\textbf{Shared}: register $R$\;
\textbf{Local}: $v$, initially $\bot$\;
\Indp
	 \textbf{procedure} $\lit{scan-validate}()$\;
	
	 \While{ $\lit{true}$ }
	 {
	    $v \gets R.\lit{read}()$;
	    $v' \gets $ new value based on $v$\;
	    $\id{flag} \gets \lit{CAS}( R, v, v')$\;
	    \If{ $\id{flag} = \lit{true}$ }
	    {
	      \textbf{output} $\id{success}$\;
	    }
	 }
\Indm
\caption{The scan-validate pattern.}
\label{algo:scan-validate}
\end{algorithm}

\subsubsection{Markov Chain Representations}
\label{section-markov-representations}

We define the \emph{extended local state} of a process in terms of the state of the system, and of the type of step it is about to take.
Thus, a process can be in one of three states: either
it performs a read, or it CAS-es with the current value of $R$, or it CAS-es with an invalid value of $R$.
The state of the system after each step is completely described by the $n$ extended local states of processes.
We emphasize that this is different than what is typically referred to as the ``local" state of a process, 
in that the extended local state is described from the viewpoint of the entire system. 
That is, a process that has a pending CAS operation can be in either of two different extended local states, depending on whether its CAS will succeed or not. 
This is determined by the state of the entire system. 
A key observation is that, although the ``local'' state of a process can only change when it takes a step, 
its extended local state can change also when another process takes a step. 

\paragraph{The individual chain.} Since the scheduler is uniform, the system can be described as a Markov chain, where each state specifies
the extended local state of each process. Specifically, a process is in state $\id{OldCAS}$ if it is about to CAS with an old (invalid) value of $R$, 
it is in state $\id{Read}$ if it is about to read,
and is in state $\id{CCAS}$ if it about to CAS with the current value of $R$. (Once CAS-ing the process returns to state $\id{Read}$.)

A state $S$ of the individual chain is given by a combination of $n$ states $S = (P_1, P_2, \ldots, P_n)$, describing the extended local state of each process,
where, for each $i \in \{1, \ldots, n\}$,   $P_i \in \{\id{OldCAS}, \id{Read}, \id{CCAS}\}$ is the extended local state of process $p_i$.
There are $3^n - 1$ possible states, since the state where each process CAS-es with an old value cannot occur.
In each transition, each process takes a step, and the state changes correspondingly.
Recall that every process $p_i$ takes a step with probability $1 / n$.
Transitions are as follows.
If the process $p_i$ taking a step is in state $\id{Read}$ or $\id{OldCAS}$, then all other processes remain in the same extended local state, and $p_i$ moves to state $\id{CCAS}$ or $\id{Read}$, respectively.
If the process $p_i$ taking a step is in state $\id{CCAS}$, then all processes in state $\id{CCAS}$ move to state $\id{OldCAS}$, and $p_i$ moves to state $\id{Read}$.


\paragraph{The system chain.} To reduce the complexity of the individual Markov chain, we introduce a simplified representation,
which tracks system-wide progress.
More precisely, each state of the system chain tracks the number of processes in each state, irrespective of their identifiers: for any $a, b \in \{0, \ldots, n\}$, a state $x$ is defined by the tuple $(a, b)$,
where $a$ is the number of processes that are in state $\id{Read}$, and $b$ is the number of processes that are in state $\id{OldCAS}$.
Notice that the remaining $n - a - b$ processes must be in state $\id{CCAS}$.
The initial state is $(n, 0)$, i.e. all processes are about to read. The state $(0, n)$ does not exist.
The transitions in the system chain are as follows.
$\Pr[ (a + 1, b - 1) | (a, b) ] = b / n, \textnormal{ where } 0 \leq a \leq n \textnormal{ and } b > 0.$
$\Pr[ (a + 1, b) | (a, b) ] = 1 - (a + b) / n, \textnormal{ where } 0 \leq a < n.$
$\Pr[ (a - 1, b) | (a, b) ] = 1 - a / n, \textnormal{ where } 0 < a \leq n.$
(See Figure~\ref{fig:chain} for an illustration of the two chains in the two-process case.)

\begin{figure}[t]
\centering
\begin{minipage}{.5\textwidth}
  \centering
  \includegraphics[scale = 0.45]{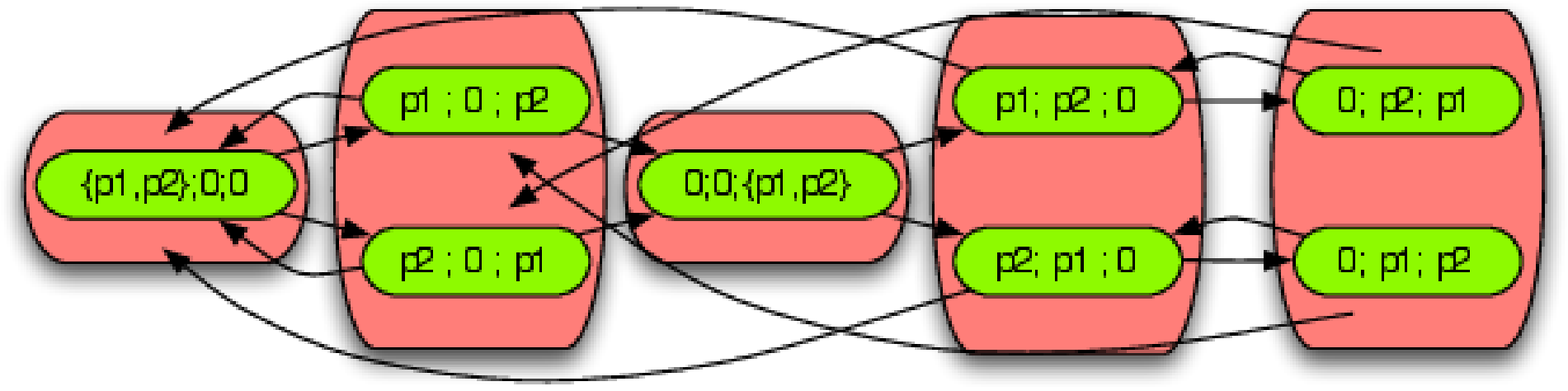}
  \captionof{figure}{The individual chain and the global chain for two processes. Each transition has probability $1/2$. The red clusters are the states in the system chain. 
  The notation $X; Y; Z$ means that processes in $X$ are in state $\id{Read}$, processes in $Y$ are in state $\id{OldCAS}$, and processes in $Z$ are in state $\id{CCAS}$.}
  \label{fig:chain}
\end{minipage}%
\begin{minipage}{.5\textwidth}
  \centering
  \includegraphics[width=.4\linewidth]{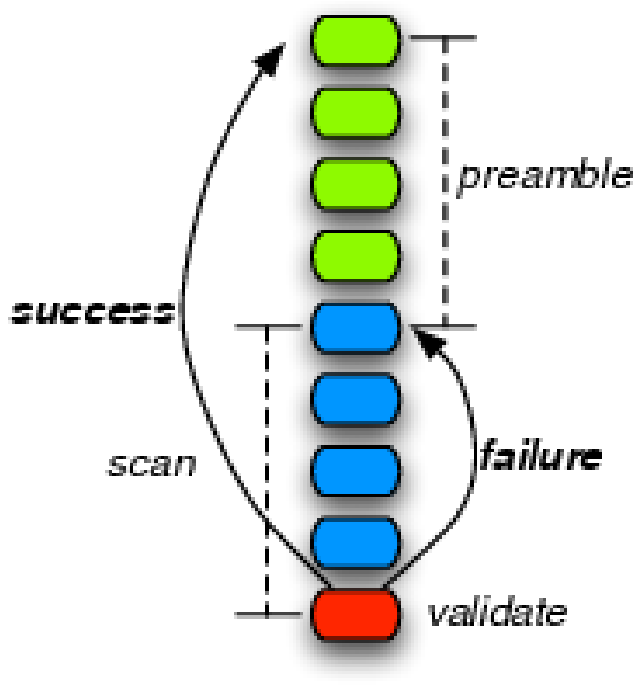}
  \captionof{figure}{Structure of an algorithm in $\SCU(q, s)$.}
  \label{fig:pattern}
\end{minipage}
\end{figure}

\subsubsection{Analysis Preliminaries}
\label{sec:analysis-appendix}

First, we notice that both the individual chain and the system chain are ergodic.

\begin{lemma}
 For any $n \geq 1$, the individual chain and the system chain are ergodic.
\end{lemma}

Let $\pi$ be the stationary distribution of the system chain, and let $\pi'$ be the stationary distribution for the individual chain.
For any state $k = (a, b)$ in the system chain, let $\pi_k$ be its probability in the stationary distribution.
Similarly, for state $x$ in the individual chain, let $\pi'_x$ be its probability in the stationary distribution.

We now prove that there exists a \emph{lifting} from the individual chain to the system chain.
Intuitively, the lifting from the individual chain to the system chain collapses all states in which
$a$ processes are about to read and $b$ processes are about to CAS with an old value (the identifiers of these processes are different for distinct states),
into to state $(a, b)$ from the system chain. 

\begin{definition}
 Let $\mathcal{S}$ be the set of states of the individual chain, and $\mathcal{M}$ be the set of states of the system chain.
 We define the function $f : \mathcal{S} \rightarrow \mathcal{M}$ such that each state $S = (P_1, \ldots, P_n)$, where
 $a$ processes are in state $\id{Read}$ and $b$ processes are in state $\id{OldCAS}$, is taken into state $(a, b)$ of the system chain.
\end{definition}

\noindent We then obtain the following relation between the stationary distributions of the two chains.
\begin{lemma}
 \label{lem:stat}
   For every state $k$ in the system chain, we have
 $ \pi_k = \sum_{x \in f^{-1}(k)} \pi'_x.$
\end{lemma}
\begin{proof}
 We obtain this relation algebraically, starting from the formula for the stationary distribution of the individual chain.
 We have that $\pi' A = \pi'$, where $\pi'$ is a row vector, and $A$ is the transition matrix of the individual chain.
 We partition the states of the individual chain into sets, where $G_{a, b}$ is the set of system states $S$ such that
 $f(S) = (a, b)$. Fix an arbitrary ordering $(G_k)_{k \geq 1}$ of the sets, and assume without loss of generality
 that the system states are ordered according to their set in the vector $\pi$ and in the matrix $A$, so that states mapping to the same set are consecutive.

 Let now $A'$ be the transition matrix across the sets $(G_k)_{k \geq 1}$. In particular, $a'_{kj}$ is the probability of moving from a state in
 the set $G_k$ to some state in the set $G_j$.
 Note that this transition matrix is the same as that of the system chain.
 Pick an arbitrary state $x$ in the individual chain, and let $f(x) = (a, b)$. In other words, state $x$ maps to set $G_k$, where $k = (a, b)$.
 We claim that for every set $G_j$,  $\sum_{y \in G_j} \Pr[ y | x ] = \Pr[ G_j | G_i ]$.

 To see this, fix $x = (P_0, P_1, \ldots, P_n)$. Since $f(x) = (a, b)$,
 there are exactly $b$ distinct states $y$ reachable from $x$ such that $f(y) = (a + 1, b - 1)$: the states where a process in extended local state $\id{OldCAS}$ takes a step.
 Therefore, the probability of moving to such a state $y$ is $b / n$.
 Similarly, the probability of moving to a state $y$ with $f(y) = (a + 1, b - 1)$ is $1 - (a + b) / n$, and the probability of moving
 to a state $y$ with $f(y) = (a - 1, b)$ is $a / n$.
  All other transition probabilities are $0$.

  To complete the proof, notice that we can collapse the stationary distribution $\pi'$ onto the row vector $\bar{\pi}$, where the $k$th element of $\bar{\pi}$
  is $\sum_{x \in G_k} \pi'_x$. Using the above claim and the fact that $\pi' A = \pi'$, we obtain by calculation that $\bar{\pi} A' = \bar{\pi}$. Therefore, $\bar{\pi}$ is
  a stationary distribution for the system chain. Since the stationary distribution is unique, $\bar{\pi} = \pi$, which concludes the proof.
\end{proof}

\noindent In fact, we can prove that the function $f : \mathcal{S} \rightarrow \mathcal{M}$ defined above induces a lifting from the individual chain to the system chain.

\begin{lemma}
 \label{lemlifting}
 The system Markov chain is a lifting of the individual Markov chain.
\end{lemma}
\begin{proof}
 Consider a state $k$ in $\mathcal{M}$. Let $j$ be a neighboring state of $k$ in the system chain.
 The ergodic flow from $k$ to $j$ is $p_{kj} \pi_{k}$. In particular, if $k$ is given by the tuple $(a, b)$,
 $j$ can be either $(a + 1, b - 1)$ or $(a + 1, b)$, or $(a - 1, b)$.
 Consider now a state $x \in \mathcal{M}$, $x = (P_0, \ldots, P_n)$, such that $f( x ) = k$.
 By the definition of $f$, $x$ has $a$ processes in state $\id{Read}$, and $b$ processes in state $\id{OldCAS}$.

 If $j$ is the state $(a + 1, b - 1)$, then the flow from $k$ to $j$, $Q_{kj}$, is $ b \pi_k  / n$.
 The state $x$ from the individual chain has exactly $b$ neighboring states $y$ which map to the state $(a + 1, b - 1)$, one for each of the $b$
 processes in state $\id{OldCAS}$ which might take a step.
 Fix $y$ to be such a state.
 The probability of moving from $x$ to $y$ is $1 / n$.
 Therefore, using Lemma~\ref{lem:stat}, we obtain that
 \begin{eqnarray*}
  \sum_{x \in f^{-1}(k), y \in f^{-1}(j)} Q'_{xy} = \sum_{x \in f^{-1}(k)} \sum_{y \in f^{-1}(j)} \frac{1}{n} \pi'_x =
  \frac{b}{n} \sum_{x \in f^{-1}(k)} \pi'_x = \frac{b}{n} \pi_k = Q_{kj}.
 \end{eqnarray*}
 \noindent The other cases for state $j$ follow similarly.
 Therefore, the lifting condition holds.
\end{proof}

\noindent Next, we notice that, since states from the individual chain which map to the same system chain state are symmetric,
their probabilities in the stationary distribution must be the same. 

\begin{lemma}
 \label{lem:symmetry}
 Let $x$ and $x'$ be two states in $\mathcal{S}$ such that $f(x) = f(y)$. Then $\pi'_x = \pi'_y$.
\end{lemma}
\begin{proof}[Proof (Sketch).] 
 The proof follows by noticing that, for any $i, j \in \{1, 2, \ldots, n\}$, switching indices $i$ and $j$ in the Markov chain representation 
 maintains the same transition matrix. Therefore, the stationary probabilities for symmetric states (under the swapping of process ids) must be the same. 
\end{proof}

\noindent We then use the fact that the code is symmetric and the previous Lemma to obtain an upper bound on the expected time between two successes
for a specific process.

\begin{lemma}
 \label{lemwork}
 Let $W$ be the expected system steps between two successes in the stationary distribution of the system chain. Let $W_i$ be the expected system steps between two successes of process $p_i$ in the stationary distribution of the individual chain.
 For every process $p_i$, $W = n W_i$.
\end{lemma}
\begin{proof}
Let $\mu$ be the probability that a step  is a success by \emph{some} process.
	Expressed in the system chain, we have that $\mu = \sum_{j = (a, b)} (1 - (a + b)/n) \pi_j$.
	Let $X_i$ be the set of states in the individual chain in which $P_i = \id{CCAS}$.
	Consider the event that a system step is a step in which $p_i$ succeeds.
	This must be a step by $p_i$ from a state in $X_i$.
	The probability of this event in the stationary distribution of the individual chain
	is $\eta_i = \sum_{x \in X_i} \pi'_x / n.$
	
	Recall that the lifting function $f$ maps all states $x$ with $a$ processes in state $\id{Read}$ and $b$ processes in state $\id{OldCAS}$
	to state $j = (a, b)$. Therefore, $\eta_i = (1 / n) \sum_{j = (a, b)} \sum_{x \in f^{-1} (j) \cap X_i} \pi'_x$.
	By symmetry, we have that $\pi'_x = \pi'_y$, for every states $x, y \in f^{-1}(j)$.
	The fraction of states in $f^{-1}(j)$ that have $p_i$ in state $\id{CCAS}$ (and are therefore also in $X_i$) is
	$(1 - (a + b)/n)$. Therefore, $\sum_{x \in f^{-1} (j) \cap X_i} \pi'_x = (1 - (a + b)/n) \pi_j$.
	
	We finally get that, for every process $p_i$, $\eta_i = (1 / n) \sum_{j = (a, b)} (1 - (a + b)/n) \pi_j = (1 / n) \mu$.
	On the other hand, since we consider the stationary distribution, from a straightforward extension of Theorem~\ref{thm:stat},
	we have that $W_i = 1 / \eta_i$, and $W = 1 / \mu$.
	Therefore, $W_i = n W$, as claimed.
\end{proof}

\subsubsection{System Latency Bound}
\label{sec:global-progress-appendix}

In this section we provide an upper bound on the quantity $W$, the expected number of system steps between two successes in stationary distribution of the system chain.
We prove the following.

\begin{theorem}
\label{thmglobal}
The expected number of steps between two successes in the system chain is $O(\sqrt{n})$.
\end{theorem}


\paragraph{An iterated balls-into-bins game.} To bound $W$, we model the evolution of the system as a balls-into-bins game.
We will associate each process with a bin. At the beginning of the execution, each bin already contains one ball.
At each time step, we throw a new ball into a uniformly chosen random bin.
Essentially, whenever the process takes a step, its bin receives an additional ball.
 We continue to distribute balls until the first time a bin acquires \emph{three} balls. We call this event a \emph{reset}.
When a reset occurs, we set the number of balls in the bin containing three balls to one,
and all the bins containing two balls become empty.
The game then continues until the next reset.

This game models the fact that initially, each process is about to read the shared state, and must take two steps in order to update its value. 
Whenever a process changes the shared state by CAS-ing successfully,
all other processes which were CAS-ing with the correct value are going to fail their operations; in particular, they now need to take three steps in order to change the shared state. We therefore reset the number of balls in the corresponding bins to $0$.

More precisely, we define the game in terms of \emph{phases}. A phase is the interval between two resets. For phase $i$,
we denote by $a_i$ the number of bins with one ball at the beginning of the phase, and by $b_i$ the number of bins with $0$ balls
at the beginning of the phase.  Since there are no bins with two or more balls at the start of a phase, we have that $a_i + b_i = n$.

It is straightforward to see that this random process evolves in the same way as the system Markov chain.
In particular, notice that the bound $W$ is the expected length of a phase. To prove Theorem~\ref{thmglobal}, we first obtain a bound on the length of a phase.

\begin{lemma}
 \label{phaselength}
  Let $\alpha \geq 4$ be a constant.
  The expected length of phase $i$ is at most $\min( 2 \alpha n / \sqrt{a_i}, 3 \alpha n / b_i^{1 / 3} )$.
  The phase length is $2 \alpha \min(n \sqrt{\log n} / \sqrt{a_i} , n (\log n)^{1 / 3} / b_i^{1/3}) ,$
  with probability at least $1 - 1 / n^\alpha$. The probability that the length of a phase is less than $\min ( n / \sqrt{a_i}, n / (b_i)^{1 / 3}) / \alpha$ is at most
 $1 / (4\alpha^2)$.
\end{lemma}
\begin{proof}
Let $A_i$ be the set of bins with one ball, and let $B_i$ be the set of bins with zero balls, at the beginning of the phase.
We have $a_i = |A_i|$ and $b_i = |B_i|$.
Practically, the phase ends either when a bin in $A_i$ or a bin in $B_i$ first contains three balls.

For the first event to occur, some bin in $A_i$ must receive two additional balls.
Let $c \geq 1$ be a large constant, and assume for now that $a_i \geq \log n$ and $b_i \geq \log n$ (the other cases will be treated separately).
The number of bins in $A_i$ which need to receive a ball before some bin receives two new balls is concentrated around $\sqrt{a_i}$, by the birthday paradox.
More precisely, the following holds.

\begin{claim}
\label{cl:aclaim}
Let $X_i$ be random variable counting the number of bins in $A_i$ chosen to get a ball before some bin in $A_i$ contains three balls,
and fix $\alpha \geq 4$ to be a constant.
Then the expectation of $X_i$ is less than $2 \alpha \sqrt{a_i}$.
The value of $X_i$ is at most $\alpha \sqrt{a_i \log n}$, with probability at least $1 - 1 / n^{\alpha^2}$.
\end{claim}
\begin{proof}
We employ the Poisson approximation for balls-into-bins processes. In essence, we want to bound the number of balls to be thrown uniformly into $a_i$ bins until
two balls collide in the same bin, in expectation and with high probability. Assume we throw $m$ balls into the $a_i \geq \log n$ bins.
It is well-known that the number of balls a bin receives during this process can be approximated as a Poisson random variable with mean $m / a_i$ (see, e.g., ~\cite{Mitz}).
In particular, the probability that no bin receives two extra balls during this process is at most

$$ 2 \left( 1 - \frac{e^{-m/a_i} (\frac{m}{a_i})^2 }{2} \right)^{a_i} \leq 2 \left(\frac{1}{e}\right)^{\frac{m^2}{2a_i} e^{-m / a_i}}.$$

\noindent If we take $m = \alpha \sqrt{a_i}$ for $\alpha \geq 4$ constant, we obtain that this probability is at most
$$ 2 \left( \frac{1}{e} \right)^{\alpha^2 e^{-\alpha / \sqrt{a_i}} / 2} \leq \left(\frac{1}{e}\right)^{\alpha^2 / 4},$$
\noindent where we have used the fact that $a_i \geq \log n \geq \alpha^2$.
Therefore, the expected number of throws until some bin receives two balls is at most $2 \alpha \sqrt{a_i}$.
Taking $m = \alpha \sqrt{a_i \log n}$, we obtain that some bin receives two new balls within $\alpha \sqrt{a_i \log n}$ throws with probability
at least $1 - 1 / n^{\alpha^2}$.
\end{proof}

\noindent We now prove a similar upper bound for the number of bins in $B_i$ which need to receive a ball before some such bin receives three new balls, as required to end the phase.

\begin{claim}
\label{cl:bclaim}
 Let $Y_i$ be random variable counting the number of bins in $B_i$ chosen to get a ball before some bin in $B_i$ contains three balls,
and fix $\alpha \geq 4$ to be a constant. Then the expectation of $Y_i$ is at most $3 \alpha b_i^{2/3}$,
and $Y_i$ is at most $\alpha (\log n)^{1 / 3} b_i^{2 / 3}$, with probability at least $1 - (1 / n)^{\alpha^3 / 54}$.
\end{claim}
\begin{proof}
 We need to bound the number of balls to be thrown uniformly into $b_i$ bins (each of which is initially empty), until some bin gets three balls.
 Again, we use a Poisson approximation.
 We throw $m$ balls into the $b_i \geq \log n$ bins.
 The probability that no bin receives three or more balls during this process is at most

 $$ 2\left( 1 - \frac{e^{-m/a_i} (m /b_i)^3 }{6} \right)^{b_i} = 2\left(\frac{1}{e}\right)^{\frac{m^3}{6b_i^{2}} e^{-m / b_i}}.$$

 \noindent Taking $m = \alpha b_i^{2/3}$ for $\alpha \geq 4$, we obtain that this probability is at most
$$ 2\left( \frac{1}{e} \right)^{\frac{\alpha^3}{6} e^{-\alpha / b_i^{1 / 3}}} \leq \left( \frac{1}{e} \right)^{\alpha^3 / 54}.$$

\noindent Therefore, the expected number of ball thrown into bins from $B_i$ until some such bin contains three balls is at most $3 \alpha b_i^{2/3}$.
Taking $m = \alpha (\log n)^{1 / 3} b_i^{2 / 3}$, we obtain that the probability that no bin receives three balls within the first
$m$ ball throws in $B_i$ is at most $(1 / n)^{\alpha^3 / 54}$.
\end{proof}

The above claims bound the number of steps inside the sets $A_i$ and $B_i$ necessary to finish the phase. On the other hand, notice that a step
throws a new ball into a bin from $A_i$ with probability $a_i / n$, and throws it into a bin in $B_i$ with probability $b_i / n$.
It therefore follows that the expected number of steps for a bin in $A_i$ to reach three balls (starting from one ball in each bin)
is at most $2 \alpha \sqrt{a_i} n / a_i = 2 \alpha n / \sqrt{a_i}$.
The expected number of steps for a bin in $B_i$ to reach three balls is at most $3 \alpha b_i^{2 / 3} n / b_i = 3 \alpha n / b_i^{1 / 3}$.
The next claim provides concentration bounds for these inequalities, and completes the proof of the Lemma.
\begin{claim}
 The probability that the system takes more than $2 \alpha \frac{n}{\sqrt{a_i}} \sqrt{\log n}$ steps in a phase
 is at most $1 / n^\alpha$.
 The probability that the system takes more than $2 \alpha \frac{n}{b_i^{1/3}} (\log n)^{1 / 3}$ steps in a phase is at most $1 / n^\alpha$.
\end{claim}
\begin{proof}
 Fix a parameter $\beta > 0$.
 By a Chernoff bound, the probability that the system takes more than $2 \beta n / a_i$ steps without throwing at least $\beta$ balls into the bins in $A_i$
 is at most $(1 / e)^{\beta}$.
 At the same time, by Claim~\ref{cl:aclaim}, the probability that $\alpha \sqrt{a_i \log n}$ balls thrown into bins in $A_i$
 do not generate a collision (finishing the phase) is at most $1 / n^{\alpha^2}$.

 Therefore, throwing $2 \alpha \frac{n}{\sqrt{a_i}} \sqrt{\log n}$ balls fail to finish the phase with probability at most
 $1 / n^{\alpha^2} + 1 / e^{\alpha \sqrt{a_i \log n}}$. Since $a_i \geq \log n$ by the case assumption, the claim follows.

 Similarly, using Claim~\ref{cl:bclaim}, the probability that the system takes more than $2 \alpha (\log n)^{1 / 3} b_i^{2 / 3} n / b_i = 2 \alpha (\log n)^{1 / 3} n / b_i^{1/3}$
 steps without a bin in $B_i$ reaching three balls (in the absence of a reset) is at most
  $(1/e)^{1 + (\log n)^{1/3} b_i^{2/3}} + (1 / n)^{\alpha^3 / 54} \leq (1 / n)^{\alpha}$, since $b_i \geq \log n$.
\end{proof}
We put these results together to obtain that, if $a_i \geq \log n$ and $b_i \geq \log n$,
then the expected length of a phase is $\min( 2 \alpha n / \sqrt{a_i}, 3 \alpha n / b_i^{1 / 3} )$.
The phase length is $2 \alpha \min(\frac{n}{\sqrt{a_i}} \sqrt{\log n}, \frac{n}{b_i^{1/3}} (\log n)^{1 / 3} ),$ with high probability.

It remains to consider the case where either $a_i$ or $b_i$ are less than $\log n$.
Assume $a_i \geq \log n$. Then $b_i \geq n - \log n$. We can therefore apply the above argument for $b_i$, and we obtain that with high probability the phase
finishes in $2 \alpha n (\log n /  b_i)^{1/3}$ steps. This is less than $2 \alpha \frac{n}{\sqrt{a_i}} \sqrt{\log n}$, since $a_i \leq \log n$,
which concludes the claim. The converse case is similar.
\end{proof}

Returning to the proof, we characterize the dynamics of the phases $i \geq 1$ based on the value of $a_i$ at the beginning of the phase.
We say that a phase $i$ is in \emph{the first range} if $a_i \in [n / 3, n]$.
Phase $i$ is in \emph{the second range} if $n/c \leq a_i < n / 3$, where $c$ is a large constant.
Finally, phase $i$ is in \emph{the third range} if $0 \leq a_i < n / c$.
Next, we characterize the probability of moving between phases.

\begin{lemma}
\label{lemranges}
	For $i \geq 1$, if phase $i$ is in the first two ranges, then the probability that phase $i + 1$ is in the third range is at most $1 / n^\alpha$.
	Let $\beta > 2c^2$ be a constant. The probability that $\beta \sqrt n$ \emph{consecutive} phases are in the third range is at most $1 / n^{\alpha}$.
\end{lemma}
\begin{proof}
We first bound the probability that a phase moves to the third range from one of the first two ranges. 
\begin{claim}
\label{cl:stay}
 For $i \geq 1$, if phase $i$ is in the first two ranges, then the probability that phase $i + 1$ is in the third range is at most $1 / n^\alpha$.
\end{claim}
\begin{proof}
 We first consider the case where phase $i$ is in range two, i.e. $n / c \leq a_i < n / 3$, and bound the probability that $a_{i + 1} < n / c$.
 By Lemma~\ref{phaselength}, the total number of system steps taken in phase $i$ is at most $2 \alpha \min(n/\sqrt{a_i} \sqrt{\log n}, n/b_i^{1/3} (\log n)^{1 / 3} ),$
  with probability at least $1 - 1 / n^\alpha$.
  Given the bounds on $a_i$, it follows by calculation that the first factor is always the minimum in this range.

  Let $\ell_i$ be the number of steps in phase $i$. Since $a_i \in [n / c, n / 3)$, the expected number of balls thrown into bins from $A_i$ is at most
  $\ell_i / 3$, whereas the expected number of balls thrown into bins from $B_i$ is at least $2\ell_i / 3$.
  The parameter $a_{i + 1}$ is $a_i$ plus the bins from $B_i$ which acquire a single ball, minus the balls from $A_i$ which acquire an extra ball.
  On the other hand, the number of bins from $B_i$ which acquire a single ball during $\ell_i$ steps is tightly concentrated around $2 \ell_i / 3$, whereas
  the number of bins in $A_i$ which acquire a single ball during $\ell_i$ steps is tightly concentrated around $\ell_i / 3$.
  More precisely, using Chernoff bounds,  given  $a_i \in [n / c, n / 3)$, we obtain that $a_i \geq a_{i + 1}$, with probability at least
  $1 - 1 / e^{\alpha \sqrt{n}}$.

  For the case where phase $i$ is in range one, notice that, in order to move to range three, the value of $a_i$ would have to decrease by at least $n ( 1 / 3 - 1 / c)$
  in this phase. On the other hand, by Lemma~\ref{phaselength}, the length of the phase is at most $2\alpha\sqrt{3n \log n}$, w.h.p.
  Therefore the claim follows.
  A similar argument provides a lower bound on the length of a phase.
\end{proof}

\noindent The second claim suggests that, if the system is in the third range (a low probability event), it gradually returns to one of the first two ranges.
\begin{claim}
\label{cl:move}
  Let $\beta > 2c^2$ be a constant. The probability that $\beta \sqrt n$ phases are in the third range is at most $1 / n^{\alpha}$.
\end{claim}
\begin{proof}
 Assume the system is in the third range, i.e. $a_i \in [ 0, n / c)$.
 Fix a  phase $i$, and let $\ell_i$ be its length.
 Let $S_b^i$ be the set of bins in $B_i$ which get a single ball during phase $i$.
 Let $T_b^i$ be the set of bins in $B_i$ which get two balls during phase $i$ (and are reset).
 Let $S_a^i$ be the set of bins in $A_i$ which get a single ball during phase $i$ (and are also reset).
 Then $b_{i} - b_{i + 1} \geq |S_b^i| - |T_b^i| - |S_a^i|$.

 We bound each term on the right-hand side of the inequality.
 Of all the balls thrown during phase $i$, in expectation at least $(1 - 1 / c)$ are thrown in bins from $B_i$.
 By a Chernoff bound, the number of balls thrown in $B_i$ is at least $(1 - 1 / c) (1 - \delta) \ell_i$ with probability
 at least $1 - \exp({-\delta^2 \ell_i  (1 - 1 / c)/ 4})$, for $\delta \in (0, 1)$.
 On the other hand, the majority of these balls do not cause collisions in bins from $B_i$.
 In particular, from the Poisson approximation, we obtain that $|S_b^i| \geq 2 |T_b^i|$ with probability at least
 $1 - (1 / n)^{\alpha + 1}$, where we have used $b_i \geq n ( 1 - 1 / c)$.

 Considering $S_a^i$, notice that, w.h.p., at most $(1 + \delta) \ell_i / c$ balls are thrown in bins from $A_i$.
 Summing up, given that  $\ell_i \geq \sqrt{n} / c$, we obtain that $b_{i} - b_{i + 1} \geq (1 - 1 / c) (1 - \delta ) \ell_i / 2 -
 (1 + \delta) \ell_i / c$,
 with probability at least $1 - \max( (1 / n)^{\alpha}, \exp(-{\delta^2 \ell_i (1 - 1 / c)/4 })$.
 For small $\delta \in (0, 1)$ and $c \geq 10$, the difference is at least $\ell_i / c^2$.
 Notice also that the probability depends on the length of the phase.

 We say that a phase is \emph{regular} if its length is at least
 $\min ( n / \sqrt{a_i}, n / (b_i)^{1 / 3}) / c$. From Lemma~\ref{phaselength}, the probability that
 a phase is regular is at least $1 - 1 / (4c^2)$. Also, in this case, $\ell_i \geq \sqrt{n} / c$, by calculation.
 If the phase is regular, then the size of $b_i$ decreases by $\Omega( \sqrt {n})$, w.h.p.

 If the phase is not regular, we simply show that, with high probability, $a_i$ does not decrease.
 Assume $a_i < a_{i + 1}$. Then, either $\ell_i < \log n$, which occurs with probability at most $1 / n^{\Omega(\log n)}$ by Lemma~\ref{phaselength},
 or the inequality $b_i - b_{i + i} \geq \ell_i / c^2$ fails, which also occurs with probability at most $1 / n^{\Omega(\log n)}$.

 To complete the proof, consider a series of $\beta \sqrt{n}$ consecutive phases, and assume that $a_i$ is in the third range for all of them.
 The probability that such a phase is regular is at least $1 - 1 / (4c^2)$, therefore, by Chernoff, a constant fraction of phases are regular, w.h.p.
 Also w.h.p., in each such phase the size of $b_i$ goes down by $\Omega(\sqrt{n})$ units.
 On the other hand, by the previous argument, if the phases are not regular, then it is still extemely unlikely that $b_i$ increases for the next phase.
 Summing up, it follows that the probability that the system stays in the third range for $\beta \sqrt{n}$ consecutive phases is at most $1 / n^{\alpha}$,
 where $\beta \geq 2c^2$, and $\alpha \geq 4$ was fixed initially.
 \end{proof}
 \noindent This completes the proof of Lemma~\ref{lemranges}. 
\end{proof}

\paragraph{Final argument.} To complete the proof of Theorem~\ref{thmglobal}, recall that we are interested in the expected length of a phase.
To upper bound this quantity, we group the states of the game according to their range as follows:
state $S_{1,2}$ contains all states $(a_i, b_i)$ in the first two ranges, i.e. with $a_i \geq n / c$.
State $S_3$ contains all states $(a_i, b_i)$ such that $a_i < n / c$.
The expected length of a phase starting from a state in $S_{1, 2}$ is $O( \sqrt{n} )$, from Lemma~\ref{phaselength}.
However, the phase length could be $\omega(\sqrt{n})$ if the state is in $S_3$. We can mitigate this fact given that
the probability of moving to range three is low (Claim~\ref{cl:stay}), and the system moves away from range three rapidly (Claim~\ref{cl:move}):
intuitively, the probability of states in $S_3$ in the stationary distribution has to be very low.

To formalize the argument, we define two Markov chains.
The first Markov chain $M$ has two states, $S_{1, 2}$ and $S_3$. The transition probability from $S_{1, 2}$ to $S_3$ is $1 / n^{\alpha}$,
whereas the transition probability from $S_3$ to $S_{1,2}$ is $x > 0 $, fixed but unknown.
Each state loops onto itself, with probabilities $1 - 1 / n^{\alpha}$ and $1 - x$, respectively.
The second
Markov chain $M'$ has two states $S$ and $R$. State $S$ has a transition to $R$, with probability $\beta \sqrt{n} / n^{\alpha}$,
and a transition to itself, with probability $1 - \beta\sqrt{n} / n^{\alpha}$.
State $R$ has a loop with probability $1 / n^{\alpha}$, and a transition to $S$, with probability $1 - 1 / n^{\alpha}$.

It is easy to see that both Markov chains are ergodic. Let  $[s \,\, r]$ be the stationary distribution of $M'$.
Then, by straightforward calculation, we obtain that $s \geq 1 - \beta \sqrt n / n^\alpha$, while $r \leq \beta \sqrt n / n^\alpha$.

On the other hand, notice that the probabilities in the transition matrix for $M'$ correspond to the probabilities in the
transition matrix for $M^{\beta \sqrt{n}}$, i.e. $M$ applied to itself $\beta \sqrt{n}$ times.
This means that the stationary distribution for $M$ is the same as the stationary distribution for $M'$.
In particular, the probability of state $S_{1, 2}$ is at least $1 - \beta \sqrt{n} / n^{\alpha}$,
and the probability of state $S_3$ is at most $\beta \sqrt{n}$.

To conclude, notice that the expected length of a phase is at most the expected length of a phase in the first Markov chain $M$.
Using the above bounds, this is at most $2 \alpha \sqrt{n}  (1 - \beta \sqrt{n} / n^{\alpha}) + \beta n^{2 / 3} \sqrt{n} / n^{\alpha} = O( \sqrt n )$, as claimed.
This completes the proof of Theorem~\ref{thmglobal}.

\subsection{Parallel Code}
\label{sec:parallel-code}

We now use the same framework to derive a convergence bound for parallel code, i.e. a method call which completes after the process
executes $q$ steps, irrespective the concurrent actions of other processes.
The pseudocode is given in Algorithm~\ref{fig:parallel}.

\setcounter{AlgoLine}{0}
\begin{algorithm}[ht]
\textbf{Shared}: register $R$\;
\Indp
	 \textbf{procedure} $\lit{call}()$
	
	 \While{ $\lit{true}$ }
	 {
	    \For{ $i$ from $1$ to $q$ }
	    {
	    	Execute $i$th step\;
		}
		\textbf{output} $\id{success}$\;
	 }
\Indm
\caption{Pseudocode for parallel code.}
\label{fig:parallel}
\end{algorithm}

\paragraph{Analysis.}
We now analyze the individual and system latency for this algorithm under the uniform stochastic scheduler. Again, we start from its Markov chain representation.
We define the individual Markov chain $M_I$ to have states $S = (C_1, \ldots, C_n)$, where $C_i \in \{0, \ldots, q - 1\}$ is the current step counter for
process $p_i$. At every step, the Markov chain picks $i$ from $1$ to $n$ uniformly at random and transitions into the state
$(C_1, \ldots, (C_i + 1) \mod q, \ldots, C_n)$. A process registers a success every time its counter is reset to $0$; the system registers a success
every time some process counter is reset to $0$.
The system latency is the expected number of system steps between two successes, and the individual latency is the expected number of system steps between two
successes by a specific process.

We now define the system Markov chain $M_S$, as follows. A state $g \in M_S$ is given by $q$ values $(v_0, v_1, \ldots, v_{q - 1})$,
where for each $j \in \{0, \ldots, q - 1\}$ $v_j$ is the number of processes with step counter value $j$, with the condition that $\sum_{j = 0}^{q - 1} v_j = n$.
Given a state $(v_0, v_1, \ldots, v_{q - 1})$, let $X$ be the set of indices $i \in \{0, \ldots, q - 1\}$ such that $v_i > 0$.
Then, for each $i \in X$, the system chain transitions into the state $(v_0, \ldots, v_i - 1, v_{i + 1} + 1, \ldots, v_{q - 1})$ with probability $v_i / n$.

It is easy to check that both $M_I$ and $M_S$ are ergodic Markov chains.
Let $\pi$ be the stationary distribution of $M_S$, and $\pi'$ be the stationary distribution of $M_I$.
We next define the mapping $f : M_I \rightarrow M_S$ which maps each state $S = (C_1, \ldots, C_n)$ to the
 state $(v_0, v_1, \ldots, v_{q - 1})$, where $v_j$ is the number of processes with counter value $j$ from $S$.
 Checking that this mapping is a lifting between $M_I$ and $M_S$ is straightforward.
 \begin{lemma}
 \label{lem:lifting-uniform}
  The function $f$ defined above is a lifting between the ergodic Markov chains $M_I$ and $M_S$.
 \end{lemma}

\noindent We then obtain bounds on the system and individual latency. 
 \begin{lemma}
 \label{lemparallel}
  For any $1 \leq i \leq n$, the individual latency for process $p_i$ is $W_i = nq$.
  The system latency is $W = q$. 
 \end{lemma}
\begin{proof}
 We examine the stationary distributions of the two Markov chains. Contrary to the previous examples, it turns out that in this case it is easier to determine
  the stationary distribution of the individual Markov chain $M_I$. Notice that, in this chain, all states have in- and out-degree $n$, and the transition probabilities are uniform
  (probability $1/ n$). It therefore must hold that the stationary distribution of $M_I$ is \emph{uniform}.
  Further, notice that a $1 / nq$ fraction of the edges corresponds to the counter of a specific process $p_i$ being reset. Therefore, for any $i$,
  the probability that a step in $M_I$ is a completed operation by $p_i$ is $1 / nq$. Hence, the individual latency for the algorithm is $nc$.
  To obtain the system latency, we notice that, from the lifting, the probability that a step in $M_S$ is a completed operation by \emph{some} process is $1 / q$.
  Therefore, the individual latency for the algorithm is $q$.
\end{proof}

\subsection{General Bound for $\SCU(q, s)$}

We now put together the results of the previous sections to obtain a bound on individual and system latency.
First, we notice that Theorem~\ref{thmglobal} can be easily extended to the case where the loop contains $s$ scan steps, as the extended local state of a process $p$ can be changed by a step of another process $q\neq p$ only if $p$ is about to perform a CAS operation.

\begin{corollary}
\label{cor:scan-validate}
 For $s \geq 1$, given a scan-validate pattern with $s$ scan steps under the stochastic scheduler, the system latency is $O( s \sqrt{n})$,
 while the individual latency is $O( n s \sqrt n )$. 
\end{corollary}

Obviously, an algorithm in $\SCU(q, s)$ is a sequential composition of parallel code followed by $s$ loop steps.
Fix a process $p_i$. By Lemma~\ref{lemparallel} and Corollary~\ref{cor:scan-validate}, by linearity of expectation,
we obtain that the expected individual latency for process $p_i$ to complete an operation is at most $n( q + \alpha s \sqrt n )$, where $\alpha \geq 4$ is a constant. 

Consider now the Markov Chain $M_S$ that corresponds to the sequential composition of the Markov chain for the parallel code $M_P$,
and the Markov chain $M_L$ corresponding to the loop. In particular, a completed operation from $M_P$ does not loop back into the chain,
but instead transitions into the corresponding state of $M_L$.
More precisely, if the transition is a step by some processor $p_i$ which completed step number $q$ in the parallel code
(and moves to the loop code), then the chain transitions into the state where processor $p_i$ is about to execute the first step of the loop code.
Similarly, when a process performs a successful CAS at the end of the loop,
the processes' step counter is reset to $0$, and its next operation will the first step of the preamble.

It is straightforward that the chain $M_S$ is ergodic.
Let $\kappa_i$ be the probability of the event that process $p_i$ completes an operation in the stationary distribution of the chain $M_S$.
Since the expected number of steps $p_i$ needs to take to complete an operation is at most $n( q + \alpha \sqrt n)$,
we have that $\kappa_i \geq 1 / (n( q + \alpha s \sqrt n))$.
Let $\kappa$ be the probability of the event that \emph{some} process completes an operation in the stationary distribution of the chain $M_S$.
It follows that $\kappa = \sum_{i = 1}^{n} \kappa_i \geq 1 / (q + \alpha s \sqrt n)$.
Hence, the expected time until the system completes a new operation is at most $q + \alpha s \sqrt n$, as claimed.

We note that the above argument also gives an upper bound on the expected number of (individual) steps a process $p_i$ 
 needs to complete an operation (similar to the standard measure of individual \emph{step complexity}). 
Since the scheduler is uniform, this is also $O( q + s \sqrt n )$.  
Finally, we note that, if only $k \leq n$ processes are correct in the execution, we obtain the same latency bounds in terms of $k$: 
since we consider the behavior of the algorithm at infinity, the stationary latencies are only influenced by correct processes. 

\begin{corollary}
\label{cor:fails}
  Given an algorithm in $\SCU(q, s)$ on $k$ correct processes under a uniform stochastic scheduler, the system latency is $O( q + s \sqrt k)$,
  and the individual latency is $O( k ( q + s \sqrt k ) )$.
\end{corollary}

\section{Application - A Fetch-and-Increment Counter using Augmented CAS}
\label{sec:applications}

We now apply the ideas from the previous section to obtain minimal and maximal progress bounds
for other lock-free algorithms under the uniform stochastic scheduler.


Some architectures support richer semantics for the CAS operation, which return the \emph{current} value of the
register which the operation attempts to modify. We can take advantage of this property to obtain a simpler
fetch-and-increment counter implementation based on compare-and-swap. This type of counter implementation is very widely-used~\cite{DiceLM13}.

\begin{algorithm}[ht]
\textbf{Shared}: register $R$\;
\Indp
	 \textbf{procedure} $\lit{fetch-and-inc}()$
	 $v \gets 0$\;
	 \While{ $\lit{true}$ }
	 {
	    $old \gets v$\;
	    $v \gets \lit{CAS}( R, v, v + 1 )$\;
	    \If{ $v = old$ }
	    {
	      \textbf{output} $\id{success}$\;
	    }
	 }
\Indm
\label{fig:fai}
\caption{A lock-free fetch-and-increment counter based on compare-and-swap.}
\end{algorithm}

\subsection{Markov Chain Representations}

We again start from the observation the algorithm induces an individual Markov chain and a global one.
From the point of view of each process, there are two possible states: \emph{Current}, in which the process has the \emph{current} value
(i.e. its local value $v$ is the same as the value of the register $R$),
and the \emph{Stale} state, in which the process has an old value, which will cause its CAS call to fail.
(In particular, the \emph{Read} and \emph{OldCAS} states from the universal construction are coalesced.)

\paragraph{The Individual Chain.} The per-process chain, which we denote by $M_I$, results from the composition of the automata
representing the algorithm at each process.
Each state of $M_I$ can be characterized by the set
of processes that have the current value of the register $R$.
The Markov chain has $2^n - 1$ states, since it never happens that
\emph{no thread} has the current value.

For each non-empty subset of processes $S$, let $s_S$ be the corresponding state.
The initial state is $s_\Pi$, the state in which every thread has the current value.
We distinguish \emph{winning} states as the states $(s_{\{p_i\}})_i$ in which only \emph{one} thread has the current value: to reach this state,
one of the processes must have successfully updated the value of $R$. There are exactly $n$ winning states, one for each process.

Transitions are defined as follows. From each state $s$, there are $n$ outgoing edges, one for each process which could be scheduled next.
Each transition has probability $1 / n$, and moves to state $s'$ corresponding to the set of processes which have the current value at the next time step.
Notice that the winning states are the only states with a self-loop, and that from every state $s_S$ the chain either moves to a state
$s_V$ with $|V| = |S| + 1$, or to a winning state for one of the threads in $S$.

\paragraph{The Global Chain.} The \emph{global} chain $M_G$ results from clustering the symmetric states states from $M_I$ into single states.
The chain has $n$ states $v_1, \ldots, v_n$, where  state $v_i$ comprises all the states $s_{S}$ in $M_G$ such that $|S| = i$.
Thus, state $v_1$ is the state in which \emph{some} process just completed a new operation.
In general, $v_i$ is the state in which $i$ processes have the current value of $R$
(and therefore may commit an operation if scheduled next).

The transitions in the global chain are defined as follows. For any $1 \leq i \leq n$, from state $v_i$ the chain moves to state $v_1$
with probability $i / n$. If $i < n$, the chain moves to state $v_{i + 1}$ with probability $1 - i / n$. Again, the state $v_1$ is the only state with a self-loop.
The intuition is that some process among the $i$ possessing the current value wins if scheduled next (and changes the current value); otherwise,
if some other thread is scheduled, then that thread will also have the current value.

\subsection{Algorithm Analysis}

We analyze the stationary behavior of the algorithm under a uniform stochastic scheduler, assuming each process invokes an infinite number of operations.

\paragraph{Strategy.} We are interested in the expected number of steps that some process $p_i$ takes between committing two consecutive operations, in the stationary distribution.
This is the \emph{individual latency}, which we denote by $W_i$.
As for the general algorithm, we proceed by first bounding the system latency $W$, which is easier to analyze, and then show that  $W_i = n W$, i.e. the algorithm is \emph{fair}.
We will use the two Markov chain representations from the previous section.
In particular, notice that $W_i$ is the expected return time of the ``win state" $v_i$ of the global chain $M_G$, and $W$ is the expected return time of the state $s_{p_i}$ in
which $p_i$ just completed an operation.

The first claim is an upper bound on the return time for $v_1$ in $M_G$.
\begin{lemma}
 The expected return time for $v_1$ is $W \leq 2 \sqrt{n} $.
\end{lemma}
\begin{proof}
 For $0 \leq i \leq n - 1$, let $Z(i)$ be the hitting time for state $v_1$ from the state where $n - i$ processes have the current value.
 In particular, $Z(0)$ is the hitting time from the state where \emph{all} processes have the correct value, and therefore $Z(0) = 1$.
 Analyzing the transitions, we obtain that $Z(i) = i Z( i - 1 ) / n + 1$. We prove that  $Z(n - 1) \leq 2\sqrt{n}$.

 We analyze two intervals: $k$ from $0$ to $n - \sqrt n$, and then up to $n - 1$.
  We first claim that, for $0 \leq k \leq n - \sqrt n$, it holds that $Z(k) \leq \sqrt n$. We prove this by induction.
  The base case obviously holds. For the induction step, notice that $Z ( k ) \leq Z ( k - 1 ) (n - \sqrt n ) / n + 1$ in this interval.
By the hypothesis, $Z(k - 1) \leq \sqrt n$, therefore $Z(k) \leq \sqrt{n}$ for $k \leq n - \sqrt n$.

For $k \in \{ n - \sqrt{n}, \ldots, n\}$, notice that $Z(k)$ can add at most $1$ at each
iteration, and we are iterating at
most $\sqrt n$ times. This gives an upper bound of $2 \sqrt n$, as claimed.
\end{proof}

\paragraph{Remark.} Intuitively, the value $Z(n - 1)$ is related to the birthday paradox, since it counts the number of elements
that must be chosen uniformly at random from $1$ to $n$ (with replacement) until one of the elements appears twice.
In fact, this is the Ramanujan $Q$ function~\cite{Flajolet}, which has been studied previously by
Knuth~\cite{Knuth3} and Flajolet et al.~\cite{Flajolet} in relation to the performance of linear probing hashing. Its asymptotics are known to be
$Z(n - 1) = \sqrt{ \pi n  / 2 } ( 1 + o(1) )$~\cite{Flajolet}.

\paragraph{Markov Chain Lifting.} We now analyze $W_i$, the expected number of total system steps for a specific process $p_i$ to commit a new request.
We define a mapping $f : M_{I} \rightarrow M_{G}$ between the states of the individual Markov chain. For any non-empty set $S$ of processes,
the function maps the state $s_S \in M_I$ to the state $v_i$ of the chain.
It is straightforward to prove that this mapping is a correct lifting of the Markov chain, and that both Markov chains are ergodic.

\begin{lemma}
\label{lem:fai-lifting}
	The individual chain and the local chain are ergodic.
	The function $f$ is a lifting between the individual chain and the global chain.
\end{lemma}

\noindent We then use the lifting and symmetry to obtain the following relation between the stationary distributions of the two Markov chains.
The proof is similar to that of Lemma~\ref{lemlifting}. This also implies that every process takes the same number of steps in expectation until completing an operation.
\begin{lemma}
	Let $\pi = [\pi_1 \ldots \pi_n]$ be the stationary distribution of the global chain, and let $\pi'$ be the stationary distribution of the individual chain.
	Let $\pi'_i$ be the probability of $s_{\{p_i\}}$ in $\pi'$. Then, for all $i \in \{1, \ldots, n\}$, $\pi'_i = \pi / n$.
	Furthermore, $W_i = n W$.
\end{lemma}

\noindent This characterizes the asymptotic behavior of the individual latency.
\begin{corollary}
	For any $i \in \{1, \ldots, n\}$, the expected number of system steps between two completed operations by process $p_i$ is $O( n \sqrt n)$.
	The expected number of steps by $p_i$ between two completed operations is $O( \sqrt n )$.
\end{corollary}

\section{Discussion}
This paper is motivated by the fundamental question of relating the theory of concurrent programming to real-world algorithm behavior.
We give a framework for analyzing concurrent algorithms which partially explains the wait-free behavior of lock-free algorithms, and their good performance in practice.
Our work is a first step in this direction, and opens the door to many additional questions.

In particular, we are intrigued by the goal of obtaining a realistic model for the unpredictable behavior of system schedulers. Even though it has some
foundation in empirical results, our uniform stochastic model is a rough approximation, and can probably be improved.
We believe that some of the elements of our framework (such as the
existence of liftings) could still be applied to non-uniform stochastic scheduler models, while others may need to be further developed.
A second direction for future work is studying other types of algorithms, and in particular implementations which export several distict methods. 
The class of algorithms we consider is \emph{universal}, i.e.,
covers any sequential object, however there may exist object implementations which do not fall in this class.
Finally, it would be interesting to explore whether there exist concurrent algorithms which avoid the $\Theta (\sqrt n)$ contention factor
in the latency, and whether such algorithms are efficient in practice.


~\\
\paragraph{Acknowledgements.} We thank George Giakkoupis, William Hasenplaugh, Maurice Herlihy, and Yuval Peres for useful discussions,
and Faith Ellen for helpful comments on an earlier version of the paper. 
\bibliographystyle{plain}
\bibliography{bib}

\appendix

\paragraph{Structure of the Appendix.} 
Section~\ref{sec:empirical} presents empirical data for the stochastic scheduler model, 
while Section~\ref{sec:results} gives compares the predicted and actual performance of an algorithm in $\SCU$. 

\section{The Stochastic Scheduler Model}
\label{sec:empirical}

\subsection{Empirical Justification}

The real-world behavior of a process scheduler arises as a complex interaction of factors
such as the timing of memory requests (influenced by the algorithm),
the behavior of the cache coherence protocol (dependent on the architecture),
or thread pre-emption (depending on the operating system).
Given the extremely complex interactions between these components, the behavior of the scheduler could be seen as \emph{non-deterministic}.
However, when recorded for extended periods of time, simple patterns emerge. Figures~\ref{fig:average} and~\ref{fig:steps}
present statistics on schedule recordings from a simple concurrent counter algorithm,  executed on a system with 16 hardware threads.
(The details of the setup and experiments are presented in the next section). 

\begin{figure}
\centering
\begin{minipage}{.42\textwidth}
  \centering
  \includegraphics[scale = 0.24]{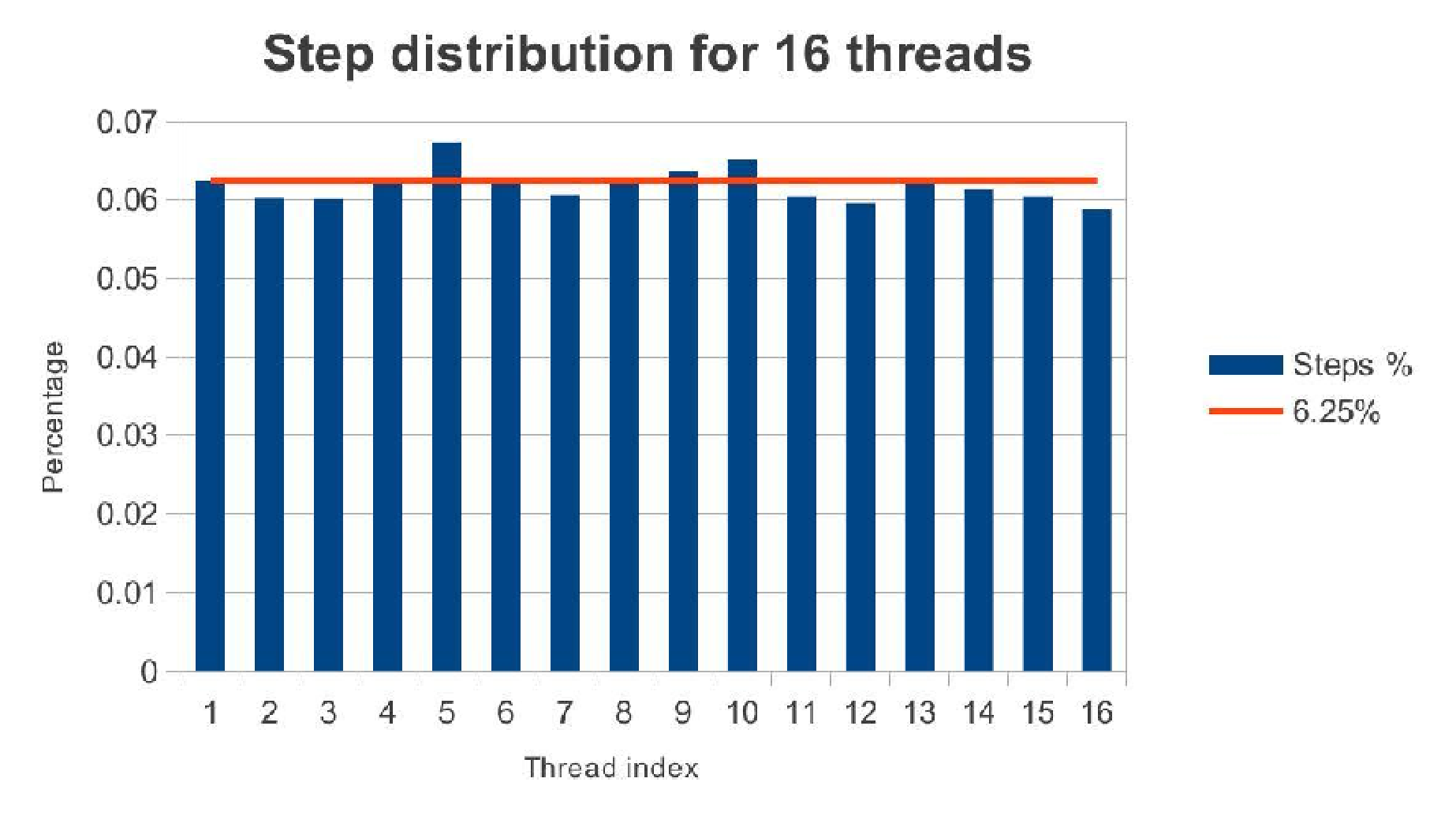}
  \captionof{figure}{Percentage of steps taken by each process during an execution.}
  \label{fig:average}
\end{minipage}%
\begin{minipage}{.42\textwidth}
  \centering
  \includegraphics[scale = 0.27]{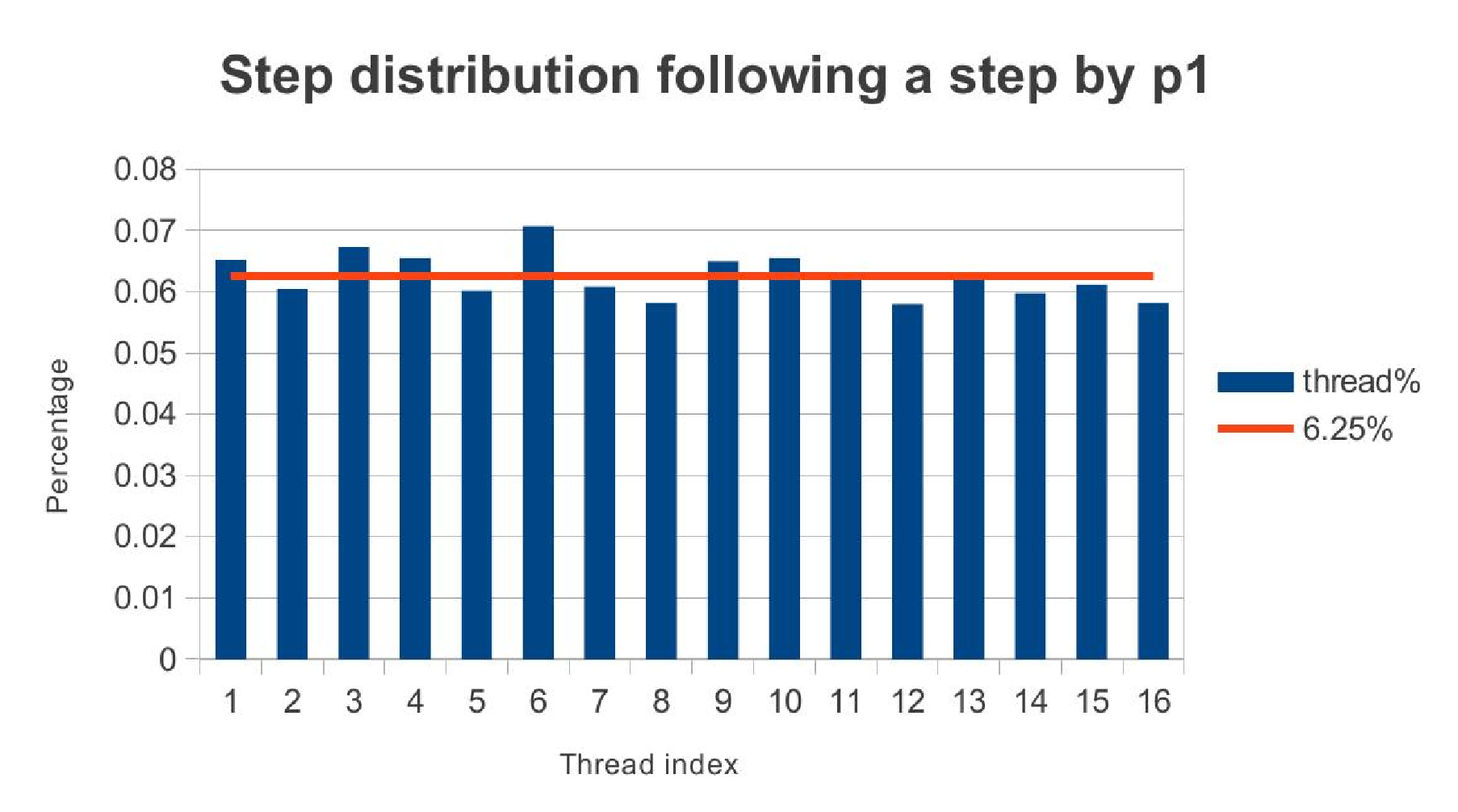}
  \captionof{figure}{Percentage of steps taken by processes, starting from a step by $p_1$. (The results are similar for all threads.)}
  \label{fig:steps}
\end{minipage}
\end{figure}

Figure~\ref{fig:average} clearly suggests that, in the long run, the scheduler is ``fair:'' each thread gets to take about the same number of steps.
Figure~\ref{fig:steps} gives an intuition about how the schedule looks like \emph{locally}: assuming process $p_i$ just took a step at time step $\tau$,
any process appears to be just as likely to be scheduled in the next step. We note that the structure of the algorithm executed can 
influence the ratios in Figure~\ref{fig:steps};  
also, we only performed tests on an Intel architecture. 

Our stochastic scheduler model addresses the non-determinism in the scheduler by associating a distribution with each scheduler time step, 
which gives the probability of each process being scheduled next.
In particular, we model our empirical observations by considering the
uniform stochastic scheduler, which assigns a probability of $1 / n$ with which each process is scheduled. 
We stress that we do not claim that the schedule behaves uniformly random locally; 
our claim is that the behavior of the schedule over long periods of time can be approximated reasonably in this way, for the algorithms we consider. 
We note that randomized schedulers attempting to explicitly implement probabilistic fairness have been proposed in practice, in the form of \emph{lottery scheduling}~\cite{Lottery}. 

\subsection{Experimental Setup}
\label{sec:exp}

The machine we use for testing is a Fujitsu PRIMERGY RX600 S6 server with four Intel Xeon
E7-4870 (Westmere EX) processors. Each processor has 10 2.40 GHz cores, each of which multiplexes two hardware threads,
so in total our system supports $80$ hardware threads. Each core has private write-back L1 and L2 caches; an inclusive
L3 cache is shared by all cores.
We limited experiments to $20$ hardware threads, in order to avoid the effects of non-uniform memory access (NUMA), which appear when hardware threads
are located on different cores.

We used two methods to record schedules. The first used an atomic fetch-and-increment operation
(available in hardware): each process repeatedly calls this operation, and records the values received.
We then sort the values of each process to recover the total order of steps.
The second method records timestamps during the execution of an algorithm, and sorts the timestamps
to recover the total order.
We found that the latter method interferes with the schedule: 
since the timer call causes a delay to the caller, a process is less likely to be scheduled twice in succession.
With this exception, the results are similar for both methods.
The statistics of the recorded schedule are summarized in Figures~\ref{fig:average} and~\ref{fig:steps}.
(The graphs are built using 20 millisecond runs, averaged over $10$ repetitions; results for longer intervals and for different thread counts are similar.)

\section{Implementation Results}
\label{sec:results}

Let the \emph{completion rate} of the algorithm be the total number of successful operations versus the total number of steps taken during the execution. 
The completion rate approximates the inverse of the system latency. 
We consider a fetch-and-increment counter implementation which simply reads the value $v$ of a shared register $R$, and then attempts to increment 
the value using a $\lit{CAS}(R, v, v + 1)$ call. 
The predicted completion rate of the algorithm is $\Theta( 1/ \sqrt n )$. The actual completion rate of the implementation is shown in Figure~\ref{fig:prediction}  
for varying thread counts, for a counter implementation based on the lock-free pattern. 
The $\Theta(1 / \sqrt n )$ rate predicted by the uniform stochastic scheduler model appears to be close to the actual completion rate. 
Since we do not have precise bounds on the constant in front of $\Theta(1 / \sqrt n)$ for the prediction, we scaled the prediction to the first data point. 
The worst-case predicted rate $(1 / n)$ is also shown.
\begin{figure}
\centering
  \includegraphics[scale = 0.4]{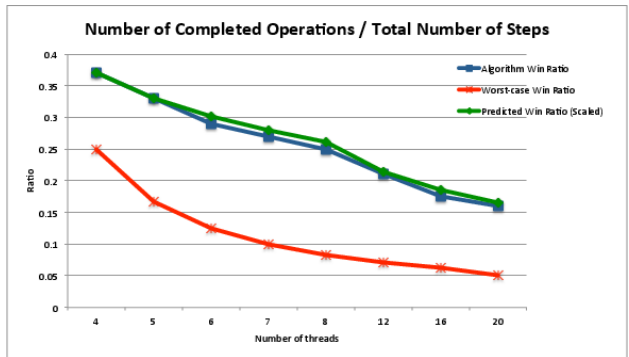}
  \captionof{figure}{Predicted completion rate of the algorithm vs. completion rate of the implementation vs. worst-case completion rate.} 
  \label{fig:prediction}
\end{figure}%

\end{document}